%% file: tech-report.tex
\documentclass[a4paper,USenglish]{llncs}

\usepackage{amsmath,amssymb}
\usepackage{mathtools}
\usepackage{xspace}
\usepackage{todonotes}
\usepackage{hyperref}
\usepackage[basic]{complexity}
\usepackage{comment}

\makeatletter
\let\epsilon\varepsilon
\let\phi\varphi

\let\rho\varrho
\makeatother

\newcommand{\PosSLP}{\textsc{PosSLP}}
\newclass{\PosSLPcx}{PosSLP}
\renewclass{\PSPACE}{PSpace}
\renewclass{\NEXP}{NExp}

\usepackage[mathlines]{lineno}
\let\oldqed\qed
\renewcommand\qed{\mbox{}\hfill$\oldqed$}
\newcommand*\patchAmsMathEnvironmentForLineno[1]{%
  \expandafter\let\csname old#1\expandafter\endcsname\csname #1\endcsname
  \expandafter\let\csname oldend#1\expandafter\endcsname\csname end#1\endcsname
  \renewenvironment{#1}%
  {\linenomath\csname old#1\endcsname}%
  {\csname oldend#1\endcsname\endlinenomath}}%
\newcommand*\patchBothAmsMathEnvironmentsForLineno[1]{%
  \patchAmsMathEnvironmentForLineno{#1}%
  \patchAmsMathEnvironmentForLineno{#1*}}%
\AtBeginDocument{%
  \patchBothAmsMathEnvironmentsForLineno{equation}%
  \patchBothAmsMathEnvironmentsForLineno{align}%
  \patchBothAmsMathEnvironmentsForLineno{flalign}%
  \patchBothAmsMathEnvironmentsForLineno{alignat}%
  \patchBothAmsMathEnvironmentsForLineno{gather}%
  \patchBothAmsMathEnvironmentsForLineno{multline}%
}

\usepackage{tikz}
\usetikzlibrary{decorations.markings, decorations.pathmorphing,
decorations.pathreplacing, arrows,automata}
\usepackage{subfig}

\newcommand{\defeq}{\overset{\mathrm{def}}{=}}

\newcommand{\st}{\mathrel{\mid}}

\newcommand{\calG}{\mathcal{G}}
\newcommand{\eve}{Eve\xspace}
\newcommand{\adam}{Adam\xspace}
\newcommand{\regret}[2]{\mathbf{Reg}\left(#2\right)}
\newcommand{\Regret}[1]{\mathbf{Reg}}
\newcommand{\Val}{\mathbf{Val}}
\newcommand{\out}[3]{\mathbf{out}^{#1}(#2,#3)}
\newcommand{\histe}[1]{\mathbf{Hist}}
\newcommand{\aVal}{\mathbf{aVal}}
\newcommand{\cVal}{\mathbf{cVal}}
\newcommand{\acVal}{\mathbf{acVal}}

\newcommand{\cOpt}{\mathbf{cOpt}}
\newcommand{\prefeq}{\mathrel{\subseteq_{\mathrm{pref}}}}
\newcommand{\sbostrat}{\sigma^{\mathrm{sbo}}}
\newcommand{\swostrat}{\sigma^{\mathrm{swo}}}
\newcommand{\sbwostrat}{\sigma^{\mathrm{sbwo}}}
\newcommand{\switch}[3]{#1\overset{#2}{\mathrm{\rightarrow}}#3}

\newcommand{\wmax}{w_{\mathrm{\max}}}
\newcommand{\resp}{resp.~\ignorespaces}
\renewenvironment{problem}{%
  \par\noindent%
  \bgroup
  \def\given{\textbf{Given:}}%
  \def\quest{\textbf{Question:}}%
    
  \begin{tabular}{@{~~}|@{~}ll}%
    \rule[1em]{0pt}{0em}\ignorespaces
    }{\rule[-1em]{0pt}{0em}%
    \\\multicolumn{2}{@{~~}l}{\rule[1.2em]{1cm}{.4pt}}\\[-1em]
  \end{tabular}
  \egroup\par
}

\begin{document}

\title{The Impatient May Use Limited Optimism\\ to Minimize Regret}

\author{Micha\"el Cadilhac\inst{1}
\and
Guillermo A. P\'erez\inst{2}
\and
Marie van den Bogaard\inst{3}}

\institute{
University of Oxford\\
\email{michael@cadilhac.name}
\and
University of Antwerp\\
\email{guillermoalberto.perez@uantwerpen.be}
\and
Universit\'e libre de Bruxelles\\
\email{marie.van.den.bogaard@ulb.ac.be}}

\maketitle

\begin{abstract}
  Discounted-sum games provide a formal model for the study of reinforcement
  learning, where the agent is enticed to get rewards early since later rewards
  are discounted.  When the agent interacts with the environment, she may regret
  her actions, realizing that a previous choice was suboptimal given the
  behavior of the environment.  The main contribution of this paper is a \PSPACE{}
  algorithm for computing the minimum possible regret of a given game. To this
  end, several results of independent interest are shown. (1) We identify a
  class of regret-minimizing and admissible strategies that first assume that
  the environment is collaborating, then assume it is adversarial---the precise
  timing of the switch is key here. (2) Disregarding the computational cost of
  numerical analysis, we provide an \NP{} algorithm that checks that the regret
  entailed by a given time-switching strategy exceeds a given value. (3) We show
  that determining whether a strategy minimizes regret is decidable in \PSPACE{}.
\end{abstract}


\section{Introduction}\label{sec:intro}
\input{intro}

\section{Preliminaries}\label{sec:preliminaries}
\input{prelim}

\section{Admissible strategies and regret}\label{sec:admissibility}
\input{admi}

\input{algos2}

\section{Conclusion}
\input{conclusions}

\subsubsection{Acknowledgements.}
We thank Rapha\"el Berton and Isma\"el Jecker for helpful conversations on the
length of maximal (and minimal) histories in discounted-sum games; and
James Worrell and Jo\"el Ouaknine for pointers on the complexity of comparing
succinctly represented integers.

\clearpage
\bibliographystyle{splncs04}
\bibliography{refs}

\clearpage
\appendix
\input{appendix}

\clearpage
\setcounter{tocdepth}{5}
\tableofcontents

\end{document}

%% file: intro.tex
A pervasive model used to study the strategies of an agent in an unknown
environment is \emph{two-player infinite horizon games played on finite weighted
  graphs}.  Therein, the set of vertices of a graph is split between two
players, Adam and Eve, playing the roles of the environment and the agent,
respectively.  The play starts in a specific vertex, and each player decides
where to go next when the play reaches one of their vertices.  Questions asked
about these games are usually of the form: \emph{Does there exist a strategy of
  Eve such that\ldots?}  For such a question to be well-formed, one should provide:
\begin{enumerate}
\item A valuation function: given an infinite play, what is Eve's reward?
\item Assumptions about the environment: is Adam trying to help or hinder Eve?
\end{enumerate}

The valuation function can be Boolean, in which case one says that Eve
\emph{wins} or \emph{loses} (one very classical example has Eve winning if the
maximum value appearing infinitely often along the edges is even).  In this
setting, it is often assumed that Adam is adversarial, and the question then
becomes: \emph{Can Eve always win?}  (The names of the players stem from this
view: \emph{is there} a strategy of \(\exists\)ve that \emph{always} beats
\(\forall\)dam?)  The literature on that subject spans more than 35 years, with newly
found applications to this day (see~\cite{ag11} for comprehensive lecture notes,
and~\cite{cgiv18} for an example of recent use in the analysis of attacks in
cryptocurrencies).

The valuation function can also aggregate the numerical values along the edges
into a reward value.  We focus in this paper on \emph{discounted sum}: if \(w\) is
the weight of the edge taken at the \(n\)-th step, Eve's reward grows by
\(\lambda^n\cdot w\), where \(\lambda \in (0,1)\) is a prescribed discount factor.  Discounting future
rewards is a classical notion used in economics~\cite{shapley53}, Markov
decision processes~\cite{puterman05,fv12}, systems theory~\cite{ahm03}, and is
at the heart of Q-learning, a reinforcement learning technique widely used in
machine learning~\cite{wd92}.  In this setting, we consider three attitudes
towards the environment:
\begin{enumerate}
\item The adversarial environment hypothesis translates to Adam trying to
  minimize Eve's reward, and the question becomes: \emph{Can Eve always achieve
    a reward of \(x\)?}  This problem is in \(\NP \cap \coNP\)~\cite{zp96} and showing
  a \(\P\) upper-bound would constitute a major breakthrough (namely, it would
  imply the same for so-called parity games~\cite{jurdzinski98}).  A strategy of
  Eve that maximizes her rewards against an adversarial environment is called
  \emph{worst-case optimal}.  Conversely, a strategy that maximizes her rewards
  assuming a \emph{collaborative} environment is called \emph{best-case
    optimal}.

\item Assuming that the environment is adversarial is drastic, if not
  pessimistic.  Eve could rather be interested in settling for a strategy
  \(\sigma\) which is not \emph{consistently} bad: if another strategy
  \(\sigma'\) gives a better reward in one environment, there should be another
  environment for which \(\sigma\) is better than \(\sigma'\).  Such strategies, called
  \emph{admissible}~\cite{bchprrs16}, can be seen as an \emph{a priori} rational
  choice.

\item Finally, Eve could put no assumption on the environment, but regret not
  having done so.  Formally, the \emph{regret value} of Eve's strategy is
  defined as the maximal difference, for all environments, between the best
  value Eve \emph{could} have obtained and the value she actually obtained.  Eve
  can thus be interested in following a strategy that achieves the minimal
  regret value, aptly called a \emph{regret-minimal} strategy~\cite{fgr10}.
  This constitutes an \emph{a posteriori} rational choice~\cite{hp12}.
  Regret-minimal strategies were explored in several contexts, with applications
  including competitive online algorithm synthesis and robot-motion
  planning~\cite{akl10,fjlpr17,hpr16,hpr17}.
\end{enumerate}

In this paper, we single out a class of strategies for Eve that first follow a
best-case optimal strategy, then switch to a worst-case optimal strategy after
some precise time; we call these strategies \emph{optipess}.  Our main
contributions are then:
\begin{enumerate}
\item Optipess strategies are not only regret-minimal (a fact established
  in~\cite{hpr16}) but also admissible---note that there are regret-minimal
  strategies that are not admissible and \emph{vice versa} (see Appendix).  On
  the way, we show that for any strategy of Eve there is an admissible strategy
  that performs at least as well; this is a peculiarity of discounted-sum
  games.
\item The regret value of a given time-switching strategy can be computed with
  an \(\NP\) algorithm (disregarding the cost of numerical analysis).  The main
  technical hurdle is showing that exponentially long paths can be represented
  succinctly, a result of independent interest.
\item The question \emph{Can Eve's regret be bounded by \(x\)?} is decidable in
  \(\NP^\coNP\), improving on the implicit \NEXP{} algorithm of~\cite{hpr16}.  The
  algorithm consists in guessing a time-switching strategy and computing its
  regret value; since optipess strategies are time-switching strategies that are
  regret-minimal, the algorithm will eventually find the minimal regret value of
  the input game.
\end{enumerate}

Notations and definitions are introduced in Section~\ref{sec:preliminaries}.
The study of admissible regret-minimal strategies is done in
Section~\ref{sec:admissibility}.  In Section~\ref{sec:minval}, we provide an
important lemma that allows to represent long paths succinctly.  In
Section~\ref{sec:witnesses}, we argue that the important values of a game
(regret, best-case, worst-case) have short witnesses.  Finally, in
Section~\ref{sec:algos}, we rely on these lemmas to present our new algorithms.

%% file: prelim.tex
We assume familiarity with basic graph and complexity theory.  Some more
specific definitions and known results are recalled here.


\paragraph{Game, play, history.}

A \emph{(discounted-sum) game} $\calG$ is a tuple
$(V,v_0, V_\exists,E,w,\lambda)$ where $V$ is a finite set of vertices, \(v_0\) is
the starting vertex, $V_\exists \subseteq V$ is the subset of vertices that
belong to \eve, $E \subseteq V \times V$ is a set of directed edges,
$w \colon E \to \mathbb{Z}$ is an (edge-)weight function, and $0 < \lambda < 1$
is a rational \emph{discount factor}.  The vertices in $V \setminus V_\exists$
are said to belong to \adam.  Since we consider games played for an infinite
number of turns, we will always assume that every vertex has at least one
outgoing edge.

A \emph{play} is an infinite path $v_1 v_2 \cdots \in V^\omega$ in the digraph
$(V,E)$.  A \emph{history} $h = v_1 \cdots v_n$ is a finite path.  The
\emph{length of \(h\)}, written \(|h|\), is the number of \emph{edges} it contains:
$|h| \defeq n - 1$.  The set $\histe{v_0}$ consists of all histories that start
in \(v_0\) and end in a vertex from $V_\exists$.

\paragraph{Strategies.}

A \emph{strategy of \eve} in $\calG$ is a function $\sigma$ that maps histories
ending in some vertex $v \in V_\exists$ to a neighbouring vertex \(v'\) (i.e.,
\((v,v') \in E\)).  The strategy $\sigma$ is \emph{positional} if for all histories
\(h, h'\) ending in the same vertex, $\sigma(h) = \sigma(h')$.  \emph{Strategies of \adam}
are defined similarly.

A history $h = v_1 \cdots v_n$ is said to be \emph{consistent with a strategy}
$\sigma$ of \eve if for all $i \geq 2$ such that $v_i \in V_\exists$, we have
that $\sigma(v_1 \cdots v_{i-1}) = v_i$.  Consistency with strategies of \adam is
defined similarly.  We write $\histe{v_0}(\sigma)$ for the set of histories in
$\histe{v_0}$ that are consistent with $\sigma$.  A play is consistent with a
strategy (of either player) if all its prefixes are consistent with it.

Given a vertex \(v\) and both Adam and Eve's strategies, \(\tau\) and
\(\sigma\) respectively, there is a unique play starting in \(v\) that is
consistent with both, called the \emph{outcome} of \(\tau\) and \(\sigma\) on
\(v\).  This play is denoted $\out{v}{\sigma}{\tau}$.

For a strategy $\sigma$ of \eve and a history $h \in \histe{v_0}(\sigma)$, we
let $\sigma_h$ be the strategy of \eve that assumes \(h\) has already been played.
Formally, \(\sigma_h(h') = \sigma(h \cdot h') \) for any history $h'$ (we will
use this notation only on histories \(h'\) that start with the ending vertex of
\(h\)).

\paragraph{Values.}

The \emph{value of a history} $h = v_1 \cdots v_n$ is the discounted sum of the
weights on the edges:
    \[
        \Val(h) \defeq \sum_{i=0}^{|h|-1} \lambda^i w(v_i,v_{i+1})\enspace.
    \]
The \emph{value of a play} is simply the limit of the values of its prefixes.

The \emph{antagonistic value} of a strategy $\sigma$ of \eve with history
$h = v_1 \cdots v_n$ is the value Eve achieves when Adam tries to hinder her,
after \(h\):
\[ \aVal^h(\sigma) \defeq \Val(h) + \lambda^{|h|} \cdot
\inf_\tau \Val(\out{v_n}{\sigma_h}{\tau})\enspace,\]
where $\tau$ ranges over all strategies of \adam.
The \emph{collaborative value} \(\cVal^h(\sigma)\) is defined in a similar way, by
substituting ``\(\sup\)'' for ``\(\inf\).''  We write \(\aVal^h\) (\resp
\(\cVal^h\)) for the best antagonistic (\resp collaborative) value achievable by
Eve with any strategy.

\paragraph{Types of strategies.}

A strategy $\sigma$ of \eve is \emph{strongly worst-case optimal} (SWO) if for
every history $h$ we have $\aVal^{h}(\sigma)= \aVal^{h}$; it is \emph{strongly
  best-case optimal}~(SBO) if for every history $h$ we have
$\cVal^{h}(\sigma)= \cVal^{h}$.

We single out a class of SWO strategies that perform well if Adam turns out to
be helping.  A SWO strategy $\sigma$ of \eve is \emph{strongly best worst-case
  optimal}~(SBWO) if for every history $h$ we have
$\cVal^{h}(\sigma)= \acVal^{h}$, where:
    \[
    \acVal^h \defeq \sup \lbrace \cVal^{h}(\sigma') \st \sigma' \text{ is a
      SWO strategy of \eve} \rbrace\enspace.
    \]


In the context of discounted-sum games, strategies that are positional and
strongly optimal always exist. Furthermore, the set of all such strategies can
be characterized by local conditions.
\begin{lemma}[Follows from~{\cite[Theorem 5.1]{zp96}}]\label{lem:charac-so}
  There exist positional SWO, SBO, and SBWO strategies in every game.  For any
  positional strategy $\sigma$ of \eve:
  \begin{itemize}
  \item $(\forall v \in V)\left[\aVal^{v}(\sigma) = \aVal^{v}\right]$ iff
    $\sigma$ is SWO;
  \item $(\forall v \in V)\left[\cVal^{v}(\sigma) = \cVal^{v}\right]$ iff
    $\sigma$ is SBO;
  \item
    \((\forall v \in V)\left[\aVal^v(\sigma) = \aVal^v \land \cVal^v(\sigma) =
    \acVal^v\right]\) iff $\sigma$ is SBWO.
  \end{itemize}
\end{lemma}

\paragraph{Regret.}

The \emph{regret} of a strategy $\sigma$ of \eve is the maximal difference
between the value obtained by using $\sigma$ and the value obtained by using an
alternative strategy:
\[
\regret{v_0}{\sigma} \defeq \sup_{\tau}\left(
  \left(\sup_{\sigma'} \Val(\out{v_0}{\sigma'}{\tau})\right) -
  \Val(\out{v_0}{\sigma}{\tau})\right)\enspace,
\]
where $\tau$ and $\sigma'$ range over all strategies of \adam and \eve,
respectively. The \emph{(minimal) regret of $\calG$} is then
$\Regret{v_0} \defeq \inf_\sigma \regret{v_0}{\sigma}$.

Regret can also be characterized by considering the point in history when Eve
should have done things differently.  Formally, for any vertices~\(u\) and
\(v\) let \( \cVal^{u}_{\lnot v}\) be the maximal \(\cVal^{u}(\sigma)\) for
strategies \(\sigma\) verifying \(\sigma(u) \neq v.\) Then:
\begin{lemma}[{\cite[Lemma 13]{hpr16}}]\label{lem:reg-by-local}
    For all strategies~$\sigma$ of \eve:
    \[
        \regret{v_0}{\sigma} = \sup
        \left\{
        \lambda^n\left(
            \cVal^{v_n}_{\lnot \sigma(h)} - \aVal^{v_n}({\sigma_h)}
        \right)
        \:\middle|\:
        h = v_0 \cdots v_n \in \histe{v_0}(\sigma)
        \right\}\enspace.
    \]
\end{lemma}

\paragraph{Switching and optipess strategies.}

Given strategies $\sigma_1,\sigma_2$ of \eve and a \emph{threshold function}
$t \colon V_\exists \to \mathbb{N} \cup \{\infty\}$, we define the
\emph{switching strategy} $\switch{\sigma_1}{t}{\sigma_2}$ for any history
$h = v_1 \cdots v_n$ ending in \(V_\exists\) as:
\[
    \switch{\sigma_1}{t}{\sigma_2}(h) = \begin{cases}
        \sigma_2(h) & \text{if } (\exists i)[i \geq t(v_i)],\\
        \sigma_1(h) & \text{otherwise.}
    \end{cases}
\]
We refer to histories for which the first condition above holds as
\emph{switched histories}, to all others as \emph{unswitched histories}.  The
strategy is said to be \emph{bipositional} if both \(\sigma_1\) and \(\sigma_2\) are
positional.  Note that in that case, if \(h\) is switched then
\(\sigma_h = \sigma_2\), and otherwise \(\sigma_h\) is the same as \(\sigma\) but with
\(t(v)\) changed to \(\max\{0, t(v) - |h|\}\) for all \(v \in V_\exists\).  In particular, if
\(|h|\) is greater than \(\max\{t(v) < \infty\}\), then \(\sigma_h\) is nearly positional: it
switches to \(\sigma_2\) as soon as it sees a vertex with \(t(v) \neq \infty\).

\label{def:optipess}%
A strategy $\sigma$ is \emph{perfectly optimistic-then-pessimistic} (optipess,
for short) if there are positional SBO and SBWO strategies $\sbostrat$ and
$\sbwostrat$ such that $\sigma = \switch{\sbostrat}{t}{\sbwostrat}$ where
\( t(v) = \inf \left\{ i \in \mathbb{N} \:\middle|\: \lambda^i \left( \cVal^v -
    \aVal^v \right) \leq \Regret{v_0} \right\}.  \)


\begin{theorem}[\cite{hpr16}]\label{thm:regret-optipess}
  For all optipess strategies $\sigma$ of \eve,
  $\regret{v_0}{\sigma} = \Regret{v_0}$.
\end{theorem}
For completeness, we give a simple proof of this result in Appendix.

\vspace{.5em}

\noindent\emph{As we have done so far, we will assume throughout the paper that
  a game \(\calG\) is fixed---with the notable exception of the results on
  complexity.}

\vspace{2.5em}

\centerline{\(\star\)}
\centerline{\(\star\quad\star\)}
\vspace{2.5em}

\begin{example}\label{exp:reg_min_not_adm}
  Consider the following game, where round vertices are owned by Eve, and square
  ones by Adam.  The double edges represent Eve's positional strategy \(\sigma\):
  \begin{center}
    \begin{tikzpicture}[->,>=stealth',shorten >=1pt,initial text={},
      el/.style={font=\scriptsize}, inner sep=1mm,auto, every
      state/.style={font=\scriptsize,minimum size=5mm,inner sep=0mm}]

      \node [initial left, state, rectangle] (v0)                     {$v_0$};
      \node [state] [above =of v0] (v1) {$v_1$};
      \node[state]          (v2) [right=of v0]         {$v_2$};
      \node[state, rectangle]          (v'1) [right=of v1]         {$v'_1$};
      \node[state, rectangle]          (v''1) [above=of v'1]         {$v''_1$};
      \node[state]          (v'2) [right=of v2]         {$v'_2$};
      \node[state]          (x) [right=of v'1]         {$x$};
      \node[state]          (y) [right=of x]         {$y$};
      
      \path (v0) edge node[above,el] {$ 0 $} (v2);
      \path (v0) edge node[above left,el] {$ 0 $} (v1);
      \path (v1) edge[double] node[above,el] {$ 0 $} (v'1);
      \path (v1) edge node[above left,el] {$ 0 $} (v''1);
      \path (v'1) edge node[above left,el] {$ 2 $} (x);
      \path (v''1) edge node[above right,el] {$ 2 $} (x);
      \path (v'1) edge[bend right] node[below,el] {$ 0 $} (y);
      
       \path (v''1) edge[bend left] node[above right,el] {$ 0 $} (y);
      
	  \path (v2) edge node[above,el] {$ 0 $} (v'2);
     
      \path (x) edge [loop right, double,el] node {$ 2 $} (x);
      \path (y) edge [loop below, double,el] node {$ 0 $} (y);
      \path (v2) edge [loop above, double,el] node {$ 0 $} (v2);
       \path (v'2) edge [loop right, double,el] node {$ 1 $} (v'2);
    \end{tikzpicture}
  \end{center}
  Eve's strategy has a regret value of $2\lambda^2/(1-\lambda)$.  This is realized when
  \adam plays from $v_0$ to $v_1$, from $v''_1$ to $x$, and from $v'_1$ to
  $y$. Against that strategy, \eve ensures a discounted-sum value of~$0$ by
  playing according to $\sigma$ while regretting not having played to $v''_1$ to
  obtain $2\lambda^2/(1-\lambda)$.\hfill\(\blacksquare\)
\end{example}


%% file: admi.tex
There is no reason for Eve to choose a strategy that is consistently worse than
another one.  This classical notion is formalized as follows:

\begin{definition}
  Let $\sigma_1, \sigma_2$ be two strategies of \eve.  We say that
  \(\sigma_1\) is \emph{weakly dominated} by $\sigma_2$ if
  $\Val(\out{v_0}{\sigma_1}{\tau}) \leq \Val(\out{v_0}{\sigma_2}{\tau})$ for
  every strategy $\tau$ of \adam.  We say that $\sigma_1$ is \emph{dominated} by
  $\sigma_2$ if $\sigma_1$ is \emph{weakly dominated} by $\sigma_2$ but not
  conversely.  A strategy $\sigma$ of \eve is \emph{admissible} if it is not
  dominated by any other strategy.
\end{definition}

\begin{example}\label{exp:adm_not_reg_min}
  Consider the following game, where the strategy \(\sigma\) of Eve is shown by the
  double edges:
  \begin{center}
    \begin{tikzpicture}[->,>=stealth',shorten >=1pt,initial text={},
el/.style={font=\scriptsize}, inner sep=1mm,auto, every
state/.style={font=\scriptsize,minimum size=5mm,inner sep=0mm}]

      \node [initial left, state] (v0)                     {$v_0$};
      \node [state,rectangle] [above right=of v0] (v1) {$v_1$};
      \node[state,rectangle]          (v2) [right=of v0]         {$v_2$};
      \node[state]          (v'1) [right=of v1]         {$v'_1$};
      \node[state]          (v''1) [above=of v'1]         {$v''_1$};
      \node[state]          (v'2) [right=of v2]         {$v'_2$};
      
      \path (v0) edge[double] node[above,el] {$ 0 $} (v2);
      \path (v0) edge node[above left,el] {$ 0 $} (v1);
      \path (v1) edge node[above,el] {$ 0 $} (v'1);
      \path (v1) edge node[above left,el] {$ 0 $} (v''1);
	  \path (v2) edge node[above,el] {$ 0 $} (v'2);
     
      \path (v'1) edge [loop right, double] node[el] {$ 10 $} (v'1);
      \path (v''1) edge [loop right, double] node[el]{$ 5 $} (v''1);
      \path (v'2) edge [loop right, double] node[el]{$ 6 $} (v'2);
           
    \end{tikzpicture}
  \end{center}
  This strategy guarantees a discounted-sum value of $6\lambda^2(1-\lambda)$
  against any strategy of \adam. Furthermore, it is worst-case optimal since
  playing to $v_1$ instead of $v_2$ would allow \adam the opportunity to ensure
  a strictly smaller value by playing to $v''_1$.  The latter also implies that
  $\sigma$ is admissible. Interestingly, playing to $v_1$ is also an admissible
  behavior of \eve since, against a strategy of \adam that plays from $v_1$ to
  $v'_1$, it obtains $10\lambda^2(1-\lambda) > 6\lambda^2(1-\lambda)$.
  \hfill\(\blacksquare\)
\end{example}

In this section, we show that (1) Any strategy is weakly dominated by an
admissible strategy; (2) Being dominated entails more regret; (3) Optipess
strategies are both regret-minimal and admissible.  
We will need the following:

\begin{lemma}[\cite{bprs16}]\label{lem:adm_char}
    A strategy $\sigma$ of \eve is admissible if and only if for every history
    $h \in \histe{v_0}(\sigma)$ the following holds: either $\cVal^{h}(\sigma) >
    \aVal^h$ or $\aVal^h(\sigma)=\cVal^h(\sigma)=\aVal^h = \acVal^h$. 
\end{lemma}

The above characterization of admissible strategies in so-called
\emph{well-formed games} was proved in~\cite[Theorem 11]{bprs16}.
Lemma~\ref{lem:adm_char} follows from the fact that discounted-sum games are
well-formed (see Appendix, Section~\ref{sec:well-formed}). 

\subsection{Any strategy is weakly dominated by an admissible strategy}

We show that discounted-sum games have the distinctive property that every
strategy is weakly dominated by an admissible strategy.  This is in stark
contrast with most cases where admissibility has been studied
previously~\cite{bprs16}.
\begin{theorem}\label{thm:fundamental-property}
  Any strategy of \eve is weakly dominated by an admissible strategy.
\end{theorem}

\begin{proof}[Sketch]
    The main idea is to construct, based on $\sigma$, a strategy $\sigma'$ that
    will switch to a SBWO strategy as soon as $\sigma$ does not satisfy the
    characterization of Lemma~\ref{lem:adm_char}.  The first argument consists
    in showing that $\sigma$ is indeed weakly dominated by $\sigma'$.  This is
    easily done by comparing, against each strategy $\tau$ of \adam, the
    values of $\sigma$ and $\sigma'$. The second
    argument consists in verifying that $\sigma'$ is indeed admissible.  This is
    done by checking that each history $h$ consistent with $\sigma'$ satisfies
    the characterization of Lemma~\ref{lem:adm_char}, that is $\cVal^h(\sigma')
    > \aVal^h$ or $\aVal^h(\sigma') = \cVal^h(\sigma')= \aVal^h = \acVal^h$.  If
    $\sigma'$ is already following an SBWO strategy at $h$, then the definition
    of SBWO strategies ensures that $\aVal^h(\sigma') = \aVal^h$ and
    $\cVal^h(\sigma') = \acVal^h$, and the part of the characterization
    satisfied only depends whether $\acVal^h > \aVal^h$.  If $\sigma'$ is still
    following $\sigma$ at $h$, the reasoning relies on the facts that $\sigma'$
    weakly dominates $\sigma$ and that $\sigma$ satisfies the characterization
    of Lemma~\ref{lem:adm_char} until $h$. This is true because $\sigma'$ and
    $\sigma$ agree up to $h$.  In the case where $\cVal^h(\sigma) >
    \aVal^h$, the weak dominance of $\sigma$ by $\sigma'$ implies that
    $\cVal^h(\sigma') \geq \cVal^h(\sigma)$ and thus that $\cVal^h(\sigma') >
    \aVal^h$.  In the case where $\aVal^h(\sigma) = \cVal^h(\sigma)= \aVal^h =
    \acVal^h$, the weak dominance of $\sigma$ by $\sigma'$ implies:
    \begin{itemize}
        \item first, that $\aVal^h(\sigma') \geq \aVal^h(\sigma)$ and thus that
            $\aVal^h(\sigma') = \aVal^h $,
        \item second, that $\cVal^h(\sigma') \geq \cVal^h(\sigma) = \acVal^h$.
    \end{itemize} 
    Since the first point shows that $\sigma'$ is worst-case optimal at $h$, we
    know, by definition of $\acVal^h$, that $\cVal^h(\sigma') \leq \acVal^h$.
    Combined with the second point, we get that $\cVal^h(\sigma') = \acVal^h$
    and thus that $\aVal^h(\sigma') = \cVal^h(\sigma')= \aVal^h = \acVal^h$.\qed
\end{proof}

\subsection{Being dominated is regretful}

\begin{theorem}\label{thm:dominance-respects-regret}
  For all strategies $\sigma,\sigma'$ of \eve such that $\sigma$ is weakly
  dominated by $\sigma'$, it holds that
  $ \regret{v_0}{\sigma'} \leq \regret{v_0}{\sigma}$.
\end{theorem}
\begin{proof}
    Let $\sigma$, $\sigma'$ be such that $\sigma$ is weakly dominated by
    $\sigma'$. This means that for every strategy $\tau$ of \adam, we have that
    $\Val(\pi) \leq \Val(\pi')$ where $\pi =
    \out{v_0}{\sigma}{\tau}$ and $\pi' = \out{v_0}{\sigma'}{\tau}$.
    Consequently, 
    we obtain
    \[
        \left(\sup_{\sigma''} \Val(\out{v_0}{\sigma''}{\tau}\right) -
        \Val(\pi') \leq
        \left(\sup_{\sigma''} \Val(\out{v_0}{\sigma''}{\tau}\right) -
        \Val(\pi).
    \]
    As
    this holds for any $\tau$, we can conclude that $\sup_{\tau}
    \sup_{\sigma''} (\Val(\out{v_0}{\sigma''}{\tau}) -
    \Val(\out{v_0}{\sigma'}{\tau})) \leq
    \sup_{\tau} \sup_{\sigma''} (\Val(\out{v_0}{\sigma''}{\tau}) -
    \Val(\out{v_0}{\sigma}{\tau}))$,
    that is $ \regret{v_0}{\sigma'} \leq  \regret{v_0}{\sigma}$.
\qed
\end{proof}

\noindent The converse of the lemma is however false.

\subsection{Optipess strategies are both regret-minimal and admissible}

Recall that there are admissible strategies that are not regret-minimal and
\emph{vice versa} (see Appendix, Section~\ref{sec:incomparable}).  However, as a
direct consequence of Theorem~\ref{thm:fundamental-property} and
Theorem~\ref{thm:dominance-respects-regret}, there always exist regret-minimal
admissible strategies.  It turns out that optipess strategies, which are
regret-minimal (Theorem~\ref{thm:regret-optipess}), are also admissible:

\begin{theorem}\label{thm:optipess_adm}
    All optipess strategies of \eve are admissible.
\end{theorem}
\begin{proof}
    Let $\sbostrat$ and $\sbwostrat$ be positional SBO and SBWO strategies of
    \eve, $\sigma$ be an optipess strategy of \eve with $\sigma =
    \switch{\sbostrat}{t}{\sbwostrat}$, and let $h = v_0 \dots v_n \in
    \histe{v_0}(\sigma)$ be a history consistent with $\sigma$.

    Suppose first that $\lambda^k\left( \cVal^{v_k} -
    \aVal^{v_k} \right) \leq \Regret{v_0}$ for some $0 \leq k \leq n$.
    That is, $h$ is a switched history and therefore $\sigma(h) = \sbwostrat(h)$
    We know that $\sigma$ will follow $\sbwostrat$ forever from $h$, thus we
    have $\cVal^h(\sigma)=\cVal^h(\sbwostrat)$ and
    $\aVal^h(\sigma)=\aVal^h(\sbwostrat)$. Recall that $\cVal^h(\sbwostrat)=
    \acVal^h$. Hence, if $\acVal^h > \aVal^h$, we have that $\cVal^h(\sigma) =
    \cVal^h(\sbwostrat) = \acVal^h > \aVal^h$, and $\sigma$ satisfies the first
    case of the characterization from Lemma~\ref{lem:adm_char}. Now, if
    $\acVal^h = \aVal^h$, the second case is satisfied: we have that
    $\cVal^h(\sigma)=\acVal^h$, and as $\aVal^h(\sigma)=\aVal^h(\sbwostrat)$,
    we also have that $\aVal^h(\sigma)= \aVal^h$ since $\sbwostrat$ is SWO.

    Suppose now that $\Regret{v_0} < \lambda^n\left( \cVal^{v_k} - \aVal^{v_k}
    \right)$ for all $0 \leq k \leq n$. By definition of optipess strategies, we
    know that $\sigma$ and $\sbostrat$ agree up to $h$, and, in particular, that
    $\sigma(h)=\sbostrat(h)$.
    Furthermore, we know that $\cVal^h > \aVal^h$, 
    otherwise at $v_n$ we have $\lambda^n(\cVal^{v_n} - \aVal^{v_n}) =
    0$, which is necessarily smaller or equal to
    $\Regret{v_0}$. Let us show that $\cVal^h(\sigma) > \aVal^h$, thus
    satisfying the first case of the characterization. Assume, towards
    contradiction, that $\cVal^h(\sigma) \leq \aVal^h$.
    Let $\tau$ be a strategy of \adam such that $h$ is consistent
    with the outcome of $\tau$ and $\sigma$, and the value of the outcome of
    $\tau$ and $\sbostrat$ is $\cVal^h$. 
    (Such strategy and outcome indeed exist because $h$ is consistent with
    $\sbostrat$ and discounted-sum value functions are continuous, see
    Appendix~\ref{sec:well-formed} for more details.) By definition of the
    regret, we have that
    $\regret{v_0}{\sigma} \geq \Val(\out{v_0}{\sbostrat}{ \tau}) -
    \Val(\out{v_0}{\sigma}{\tau})$.  We already know that
    $\Val(\out{v_0}{\sbostrat}{ \tau}) = \cVal^h$ and $\Val(\out{v_0}{\sigma}{
    \tau}) \leq \cVal^h(\sigma) \leq \aVal^h$.  Thus,
    $\Val(\out{v_0}{\sbostrat}{ \tau}) - \Val(\out{v_0}{\sigma}{ \tau}) \geq
    \cVal^h - \aVal^h$, that is $\regret{v_0}{\sigma} \geq  \cVal^h - \aVal^h$.
    On the other hand, since the strategy $\sigma$ is regret-minimizing, it
    holds that $\regret{v_0}{\sigma} = \Regret{v_0}$. Hence,
    $\regret{v_0}{\sigma} < \lambda^n\left( \cVal^{v_n} - \aVal^{v_n} \right)$.
    But we also have $\cVal^h - \aVal^h = (\Val(h) +\lambda^n \cVal^{v_n}) -
    (\Val(h) + \lambda^n \aVal^{v_n}) = \lambda^n\left( \cVal^{v_n} -
    \aVal^{v_n} \right)$. We thus get a contradiction: $\regret{v_0}{\sigma} <
    \lambda^n\left( \cVal^{v_n} - \aVal^{v_n} \right)$ and $\regret{v_0}{\sigma}
    \geq \lambda^n\left( \cVal^{v_n} - \aVal^{v_n} \right)$.
\qed
\end{proof}

%
%
%


%% file: algos2.tex
\section{Minimal values are witnessed by a single iterated cycle}\label{sec:minval}

We start our technical work towards a better algorithm to compute the regret
value of a game.  In this section, we show a crucial lemma on representing long
histories: there are histories of a simple shape that witness small values in
the game.

More specifically, we show that for any history \(h\), there is another history
\(h'\)  of
the same length that has smaller value and such that
\(h' = \alpha \cdot \beta^k \cdot \gamma\) where \(|\alpha\beta\gamma|\) is small.  This will allow us to find the
smallest possible value among exponentially large histories by guessing
\(\alpha, \beta, \gamma,\) and \(k\), which will all be small.  This property holds for a wealth
of different valuation functions, hinting at possible further applications.
Namely, the only requirement is the following:

\begin{lemma}\label{lem:replace-cycle}
  For any history $h = \alpha \cdot \beta \cdot \gamma$ with $\alpha$ and $\gamma$ same-length cycles:
  \[\min \{\Val(\alpha^2 \cdot \beta),\Val(\beta \cdot \gamma^2)\} \leq \Val(h)\enspace.\]
\end{lemma}

Within the proof of the key lemma of this section, and later on when we use it
(Lemma~\ref{lem:aval}), we will rely on the following elementary notion of
cycle decomposition:
\begin{definition}
  A \emph{simple-cycle decomposition} (SCD) is a pair consisting of paths and
  iterated simple cycles. Formally, an SCD is a pair
  \(D = \langle (\alpha_i)_{i=0}^n, (\beta_j,k_j)_{j=1}^n\rangle\), where each
  \(\alpha_i\) is a path, each \(\beta_j\) is a simple cycle, and each \(k_j\) is a positive
  integer.  We write \(D(j) = \beta_j^{k_j}\cdot \alpha_j\) and \(D(\star) = \alpha_0\cdot D(1)D(2) \cdots D(n)\).
\end{definition}

\noindent By carefully iterating Lemma~\ref{lem:replace-cycle}, we have:
\begin{lemma}\label{lem:small-witness}
  For any history \(h\) there exists an history \(h' = \alpha \cdot \beta^k \cdot \gamma\) with:
  \begin{itemize}
  \item \(h\) and \(h'\) have the same starting and ending vertices, and the
    same length;
  \item $\Val(h') \leq \Val(h)$,
  \item $|\alpha\beta\gamma| \leq 4|V|^3$ and $\beta$ is a simple cycle.
  \end{itemize}
\end{lemma}
\begin{proof}
  \def\SSCDs{\ss{}CDs\xspace}%
  \def\SSCD{\ss{}CD\xspace}%
  In this proof, we focus on SCDs for which each path \(\alpha_i\) is simple; we call
  them \SSCDs.  We define a wellfounded partial order on \SSCDs.  Let
  $D = \langle (\alpha_i)_{i=0}^n, (\beta_j,k_j)_{j=1}^n\rangle$ and
  $D' = \langle (\alpha'_i)_{i=0}^{n'}, (\beta'_j,k'_j)_{j=1}^{n'}\rangle$ be two \SSCDs; we write
  \(D' < D\) iff all the following holds:
  \begin{itemize}
  \item \(D(\star)\) and \(D'(\star)\) have the same starting and ending vertices, the same
    length, and satisfy \(\Val(D'(\star)) \leq \Val(D(\star))\) and \(n' \leq n\);
  \item Either \(n' < n\), or \(|\alpha_0'\cdots\alpha_{n'}'| < |\alpha_0\cdots\alpha_n|\),
    or \(|\{k_i' \geq |V|\}| < |\{k_i \geq |V|\}|\).
  \end{itemize}

  \noindent%
  That this order has no infinite descending chain is clear.  We show two claims:
  \begin{enumerate}
  \item Any \SSCD with \(n\) greater than \(|V|\) has a smaller \SSCD;
  \item Any \SSCD with two \(k_{j}, k_{j'} > |V|\) has a smaller \SSCD.
  \end{enumerate}

  Together they imply that for a smallest \SSCD \(D\), \(D(\star)\) is of the required
  form.  Indeed let \(j\) be the unique value for which \(k_j > |V|\), then the
  statement of the Lemma is satisfied by letting
  \(\alpha = \alpha_0\cdot D(1)\cdots D(j-1)\), \(\beta = \beta_j\), \(k = k_j\), and \(\gamma = \alpha_j\cdot D(j+1)\cdots D(n)\).

  \emph{Claim 1.}\quad Suppose \(D\) has \(n > |V|\). Since all cycles are simple, there
  are two cycles \(\beta_j, \beta_{j'}\), \(j < j'\), of same length.  We can apply
  Lemma~\ref{lem:replace-cycle} on the path
  \(\beta_j\cdot(\alpha_jD(j+1)\cdots D(j'-1))\cdot\beta_{j'}\), and remove one of the two cycles while
  duplicating the other; we thus obtain a similar path of smaller value.  This
  can be done repeatedly until we obtain a path with only one of the two cycles,
  say \(\beta_{j'}\), the other case being similar.  Substituting this path in
  \(D(\star)\) results in:
  \[
      \alpha_0\cdot D(1)\cdots D(j) \cdot \left( \alpha_{j} \cdot D(j+1) \cdots
      D(j'-1)\cdot \beta_{j'}^{k_j+k_{j'}} \right) \cdot \alpha_{j'} \cdot
      D(j'+1)\cdots D(n)\enspace.
  \]%
  This gives rise to a smaller \SSCD as follows.  If \(\alpha_{j-1}\alpha_j\) is still a
  simple path, then the above history is expressible as an \SSCD with a smaller
  number of cycles.  Otherwise, we rewrite
  \(\alpha_{j-1}\alpha_j = \alpha'_{j-1}\beta_j'\alpha'_j\) where
  \(\alpha'_{j-1}\) and \(\alpha_j'\) are simple paths and \(\beta_j'\) is a simple cycle; since
  \(|\alpha'_{j-1}\alpha'_j| < |\alpha_{j-1}\alpha_j|\), the resulting \SSCD is smaller.

  \emph{Claim 2.}\quad Suppose \(D\) has two \(k_{j}, k_{j'} > |V|\),
  \(j < j'\).  Since each cycle in the \SSCD is simple, \(k_{j} \) and \(k_{j'}\) are
  greater than both \(|\beta_j|\) and \(|\beta_{j'}|\); let us write \(k_j = b|\beta_{j'}| + r\)
  with \(0 \leq r < |\beta_{j'}|\), and similarly, \(k_{j'} = b'|\beta_{j}| + r'\).  We have:
  \[D(j)\cdots D(j') = \beta_j^r\cdot \left((\beta_j^{|\beta_{j'}|})^b \cdot \alpha_j \cdot D(j+1)\cdots
    D(j'-1)\cdot(\beta_{j'}^{|\beta_{j}|})^{b'}\right) \cdot
  \beta_{j'}^{r'}\cdot\alpha_{j'}\enspace.\]%
  Noting that \(\beta_{j'}^{|\beta_{j}|}\) and \(\beta_{j}^{|\beta_{j'}|}\) are cycles of the same
  length, we can transfer all the occurrences of one to the other, as in
  Claim~1.  Similarly, if two simple paths get merged and give rise to a cycle,
  a smaller \SSCD can be constructed; if not, then there are now at most
  \(r <|V|\) occurrences of \(\beta_{j'}\) (or conversely, \(r'\) of
  \(\beta_j\)), again resulting in a smaller \SSCD.\qed
\end{proof}

\section{Short witnesses for regret, antagonistic, and collaborative
  values}\label{sec:witnesses}

We continue our technical work towards our algorithm for computing the regret
value.  In this section, the overarching theme is that of \emph{short
  witnesses}.  We show that (1)~The regret value of a strategy is witnessed by
histories of bounded length; (2)~The collaborative value of a game is witnessed
by a simple path and an iterated cycle; (3)~The antagonistic value of a strategy
is witnessed by an SCD and an iterated cycle.

\subsection{Regret is witnessed by histories of bounded length}

\begin{lemma}\label{lem:sup2max}
  Let \(C = 2|V| + \max\{t(v) < \infty\}\).  For any bipositional switching
  strategy $\sigma$ of \eve, we have:
  \begin{align*}
    \regret{v_0}{\sigma} = \max \Big\{
    \lambda^n\left(
    \cVal^{v_n}_{\lnot \sigma(h)} - \aVal^{v_n}(\sigma_h)
    \right) \Big| \qquad\qquad\qquad\qquad \\
    h = v_0 \dots v_n \in \histe{v_0}(\sigma), n \leq C\Big\}    \enspace.
  \end{align*}
\end{lemma}
\begin{proof}
  Consider a history \(h\) of length greater than \(C\), and write
  \(h = h_1\cdot h_2\) with \(|h_1| = \max\{t(v) < \infty\}\).  Let \(h_2 = p\cdot p'\) where \(p\) is the maximal prefix
  of \(h_2\) such that \(h_1\cdot p\) is unswitched---we set \(p = \epsilon\) if
  \(h\) is switched.  Note that one of \(p\) or
  \(p'\) is longer than \(|V|\)---say \(p\), the other case being similar.  This implies
  that there is a cycle in \(p\), i.e., \(p = \alpha\cdot\beta\cdot\gamma\) with
  \(\beta\) a cycle.  Let \(h' = h_1\cdot\alpha\cdot\gamma\cdot p'\); this history has the same starting and
  ending vertex as \(h\).  Moreover, since \(|h_1|\) is larger than any value of
  the threshold function, \(\sigma_h = \sigma_{h'}\).  Lastly, \(h'\) is still in
  \(\histe{v_0}(\sigma)\), since the removed cycle did not play a role in switching
  strategy.  This shows:
  \[\cVal^{v_n}_{\lnot \sigma(h)} - \aVal^{v_n}(\sigma_h) = \cVal^{v_n}_{\lnot \sigma(h')} -
  \aVal^{v_n}(\sigma_{h'})\enspace.\]

  Since the length of \(h\) is greater than the length of \(h'\), the discounted
  value for \(h'\) will be greater than that of \(h\), resulting in a bigger regret
  value.  There is thus no need to consider histories of size greater than \(C\).  \qed
\end{proof}

It may seem from this lemma and the fact that \(t(v)\) may be very large that we
will need to guess histories of important length.  However, since we will be
considering bipositional switching strategies, we will only be interested in
\emph{some} properties of the histories that are not hard to verify:
\begin{lemma}\label{lem:shorth}
  The following problem is decidable in \NP:
  \begin{problem}
    \given & A game, a bipositional switching strategy \(\sigma\), \\
           & a number \(n\) in binary, a Boolean \(b\), and two vertices \(v, v'\)\\
           \quest & Is there a \(h \in \histe{v_0}(\sigma)\) of length \(n\), switched if \(b\),\\
           & ending in \(v\), with \(\sigma(h) = v'\)?
                       \end{problem}
\end{lemma}
\begin{proof}
  This is done by guessing multiple flows within the graph \((V, E)\).  Here, we
  call \emph{flow} a valuation of the edges \(E\) by integers, that describes the
  number of times a path crosses each edge.  Given a vector in \(\mathbb{N}^E\),
  it is not hard to check that there is a path that it represents, and to
  extract the initial and final vertices of that path~\cite{reutenauer90}.

  We first order the different thresholds from the strategy
  \(\sigma = \switch{\sigma_1}{t}{\sigma_2}\): let
  \(V_\exists = \{v_1, v_2, \ldots, v_k\}\) with \(t(v_i) \leq t(v_{i+1})\) for all
  \(i\).  We analyze the structure of histories consistent with \(\sigma\).  Let
  \(h \in \histe{}(\sigma)\), and write \(h = h'\cdot h''\) where \(h'\) is the maximal
  unswitched prefix of \(h\).  Naturally, \(h'\) is consistent with \(\sigma_1\) and
  \(h''\) is consistent with \(\sigma_2\).  Then \(h' = h_0h_1\cdots h_i\), for some
  \(i < |V_\exists|\), with:
  \begin{itemize}
  \item \(|h_0| = t(v_1)\) and for all \(1 \leq j < i\), \(|h_j| = t(v_{j+1}) - t(v_j)\);
  \item For all \(0 \leq j \leq i\), \(h_j\) does not contain a vertex \(v_k\) with \(k \leq j\).
  \end{itemize}

  To check the existence of a history with the given parameters, it is thus
  sufficient to guess the value \(i \leq |V_\exists|\), and to guess \(i\) connected flows
  (rather than paths) with the above properties that are consistent with
  \(\sigma_1\).  Finally, we guess a flow for \(h''\) consistent with
  \(\sigma_2\) if we need a switched history, and verify that it is starting at a
  switching vertex. The flows must sum to \(n+1\), with the last vertex being
  \(v'\), and the previous~\(v\).\qed
\end{proof}

\subsection{Short witnesses for the collaborative and antagonistic values}

\begin{lemma}\label{lem:cval}
  There is a set \(P\) of pairs \((\alpha, \beta)\) with \(\alpha\) a simple path and \(\beta\) a simple
  cycle such that:
  \[\cVal^{v_0} = \max\{\Val(\alpha\cdot\beta^\omega) \mid (\alpha, \beta) \in P\}\enspace.\]
  Additionally, membership in \(P\) is decidable in polynomial time w.r.t.\ the game.
\end{lemma}
\begin{proof}
  This is a consequence of Lemma~\ref{lem:charac-so}: Consider positional SBO
  strategies \(\tau\) and \(\sigma\) of Adam and Eve, respectively.  Since they are
  positional, the path \(\out{v_0}{\sigma}{\tau}\) is of the form
  \(\alpha \cdot \beta^\omega\), as required, and its value is \(\cVal^{v_0}\).

  Moreover, it can be easily checked that, given a pair \((\alpha, \beta)\), there exists a
  pair of strategies with outcome \(\alpha\cdot\beta^\omega\).  If that holds, the value
  \(\Val(\alpha\cdot\beta^\omega)\) will be at most \(\cVal^{v_0}\).\qed
\end{proof}

\begin{lemma}\label{lem:aval}
  Let \(\sigma\) be a bipositional switching strategy of Eve.  There is a set
  \(K\) of pairs~\((D, \beta)\) with \(D\) an SCD and \(\beta\) a simple cycle such that:
  \[\aVal^{v_0}(\sigma) = \min \{\Val(D(\star)\cdot \beta^\omega) \mid (D, \beta) \in K\}\enspace.\]
  Additionally, the size of each pair is polynomially bounded, and membership in
  \(K\) is decidable in polynomial time w.r.t.\ \(\sigma\) and the game.
\end{lemma}
\begin{proof}
  Let \(C = \max \{t(v) < \infty\}\), and consider a play \(\pi\) consistent with
  \(\sigma\) that achieves the value \(\aVal^{v_0}(\sigma)\).  Write
  \(\pi = h\cdot\pi'\) with \(|h| = C\), and let \(v\) be the final vertex of
  \(h\).  Naturally:
  \[\aVal^{v_0}(\sigma) = \Val(\pi) = \Val(h) + \lambda^{|h|} \Val(\pi')\enspace.\]

  We first show how to replace \(\pi'\) by some \(\alpha\cdot\beta^\omega\), with
  \(\alpha\) a simple path and \(\beta\) a simple cycle.  First, since \(\pi\) witnesses
  \(\aVal^{v_0}(\sigma)\), we have that \(\Val(\pi') = \aVal^v(\sigma_h)\).  Now
  \(\sigma_h\) is positional, because \(|h| \geq C\).\footnote{Technically,
    \(\sigma_h\) is positional in the game where we record whether the switch was
    made.}  It is known that there are optimal positional antagonistic
  strategies \(\tau\) for Adam, that is, that satisfy
  \(\aVal^v(\sigma_h) = \out{v}{\sigma_h}{\tau}\).  As in the proof of Lemma~\ref{lem:cval},
  this implies that
  \(\aVal^v(\sigma_h) = \Val(\alpha\cdot\beta^\omega) = \Val(\pi')\) for some
  \(\alpha\) and \(\beta\); additionally, any \((\alpha, \beta)\) that are consistent with
  \(\sigma_h\) and a potential strategy for Adam will give rise to a bigger value.

  We now argue that \(\Val(h)\) is witnessed by an SCD of polynomial size.  This
  bears similarity to the proof of Lemma~\ref{lem:shorth}.  Specifically, we
  will reuse the fact that histories consistent with \(\sigma\) can be split into
  histories played ``between thresholds.''

  Let us write \(\sigma = \switch{\sigma_1}{t}{\sigma_2}\).  Again, we let
  \(V_\exists = \{v_1, v_2, \ldots, v_k\}\) with \(t(v_i) \leq t(v_{i+1})\) for all
  \(i\) and write \(h = h'\cdot h''\) where \(h'\) is the maximal unswitched prefix of
  \(h\).  We note that \(h'\) is consistent with \(\sigma_1\) and \(h''\) is consistent with
  \(\sigma_2\).  Then \(h' = h_0h_1\cdots h_i\), for some \(i < |V_\exists|\), with:
  \begin{itemize}
  \item \(|h_0| = t(v_1)\) and for all \(1 \leq j < i\), \(|h_j| = t(v_{j+1}) - t(v_j)\);
  \item For all \(0 \leq j \leq i\), \(h_j\) does not contain a vertex \(v_k\) with \(k \leq j\).
  \end{itemize}

  We now diverge from the proof of Lemma~\ref{lem:shorth}.  We apply
  Lemma~\ref{lem:small-witness} on each \(h_j\) in the game where the strategy
  \(\sigma_1\) is hardcoded (that is, we first remove every edge
  \((u, v) \in V_\exists \times V\) that does not satisfy \(\sigma_1(u) = v\)).  We obtain a history
  \(h_0'h_1'\cdots h_i'\) that is still in \(\histe{}(\sigma)\), thanks to the previous
  splitting of \(h\).  We also apply Lemma~\ref{lem:small-witness} to \(h'\), this
  time in the game where \(\sigma_2\) is hardcoded, obtaining \(h''\).  Since each
  \(h_j'\) and \(h''\) are expressed as \(\alpha\cdot\beta^k\cdot\gamma\), there
  is an SCD
  \(D\) with no more than \(|V_\exists|\) elements that satisfies
  \(\Val(D(\star)) \leq \Val(h)\)---naturally, since \(\Val(h)\) is minimal and
  \(D(\star) \in \histe{}(\sigma)\), this means that the two values are equal.  Note that it
  is not hard, given an SCD \(D\), to check whether
  \(D(\star) \in \histe{}(\sigma)\), and that SCDs that are not valued
  \(\Val(h)\) have a bigger value.\qed
\end{proof}

\section{The complexity of regret}\label{sec:algos}

We are finally equipped to present our algorithms.  To account for the cost of
numerical analysis, we rely on the problem \PosSLP~\cite{abkm09}.  This problem
consists in determining whether an arithmetic circuit with addition,
subtraction, and multiplication gates, together with input values, evaluate to a
positive integer. \PosSLP~is known to be decidable in the so-called counting
hierarchy, itself contained in the set of problems decidable using polynomial
space.

\begin{theorem}\label{thm:compute}
  The following problem is decidable in \(\NP^\PosSLPcx\):
  \begin{problem}
    \given & A game, a bipositional switching strategy \(\sigma\), \\
           & a value \(r \in \mathbb{Q}\) in binary\\
    \quest & Is \(\regret{v_0}{\sigma} > r\)?
  \end{problem}  
\end{theorem}
\begin{proof}
  \def\tc{{\text{c}}}\def\ta{{\text{a}}}
  Let us write \(\sigma = \switch{\sigma_1}{t}{\sigma_2}\).  Lemma~\ref{lem:sup2max} indicates
  that \(\regret{v_0}{\sigma} > r\) holds if there is a history \(h\) of some length
  \(n \leq C = 2|V| + \max\{t(v) < \infty\}\), ending in some \(v_n\) such that:
  \begin{align}\label{eq:gr}
    \lambda^n\left(
    \cVal^{v_n}_{\lnot \sigma(h)} - \aVal^{v_n}(\sigma_h)
    \right)  > r\enspace.  
  \end{align}
  Note that since \(\sigma\) is bipositional, we do not need to know everything
  about~\(h\). Indeed, the following suffice: its length \(n\), final vertex \(v_n\),
  \(v' = \sigma(h)\), and whether it is switched.  Rather than guessing \(h\), we can
  thus rely on Lemma~\ref{lem:shorth} to get the required information.  We start
  by simulating the NP machine that this lemma provides, and verify that
  \(n, v_n,\) and \(v\) are consistent with a potential history.

  Let us now concentrate on the collaborative value that we need to evaluate in
  Equation~\ref{eq:gr}.  To compute \(\cVal\), we rely on Lemma~\ref{lem:cval},
  which we apply in the game where \(v_n\) is set initial, and its successor
  forced not to be \(v\).  We guess a pair \((\alpha_\tc, \beta_\tc) \in P\); we thus have
  \(\Val(\alpha_\tc\cdot\beta_\tc^\omega) \leq \cVal^{v_n}_{\lnot \sigma(h)}\), with at least one guessed
  pair \((\alpha_\tc, \beta_\tc)\) reaching that latter value.

  Let us now focus on computing \(\aVal^{v_n}(\sigma_h)\).  Since \(\sigma\) is a bipositional
  switching strategy, \(\sigma_h\) is simply \(\sigma\) where \(t(v)\) is changed to
  \(\max\{0, t(v) - n\}\).  Lemma~\ref{lem:aval} can thus be used to compute our
  value.  To do so, we guess a pair \((D, \beta_\ta) \in K\); we thus have
  \(\Val(D(\star)\cdot\beta_\ta^\omega) \geq \aVal^{v_n}(\sigma_h)\), and at least one pair
  \((D, \beta_\ta)\) reaches that latter value.

  Our guesses satisfy:
  \[\cVal^{v_n}_{\lnot \sigma(h)} -
  \aVal^{v_n}(\sigma_h) \geq \Val(\alpha_\tc\cdot\beta_\tc^\omega) - \Val(D(\star)\cdot\beta_\ta^\omega) \enspace,\]%
  and there is a choice of our guessed paths and SCD that gives exactly the
  left-hand side.  Comparing the left-hand side with \(r\) can be done using an
  oracle to \PosSLP\ (see Appendix, Section~\ref{sec:posslp}), concluding the
  proof.  \qed
\end{proof}

\begin{theorem}\label{thm:reg}
  The following problem is decidable in \(\coNP^{\NP^\PosSLPcx}\):
  \begin{problem}
    \given & A game, a value \(r \in \mathbb{Q}\) in binary\\
    \quest & Is \(\Regret{v_0} > r\)?
  \end{problem}
\end{theorem}
\begin{proof}
  To decide the problem at hand, we ought to check that \emph{every} strategy
  has a regret value greater than \(r\).  However, optipess strategies being
  regret-minimal, we need only check this for a class of strategies that
  contains optipess strategies: bipositional switching strategies form one such
  class.

  What is left to show is that optipess strategies can be encoded in
  \emph{polynomial space}.  Naturally, the two positional strategies contained
  in an optipess strategy can be encoded succinctly.  We thus only need to show
  that, with \(t\) as in the definition of optipess strategies
  (page~\pageref{def:optipess}), \(t(v)\) is at most exponential for every
  \(v \in V_\exists\) with $t(v) \in \mathbb{N}$.  This is shown in Appendix,
  Section~\ref{sec:boundopti}.\qed
\end{proof}

\begin{theorem}
  The following problem is decidable in \(\coNP^{\NP^\PosSLPcx}\):
  \begin{problem}
    \given & A game, a bipositional switching strategy \(\sigma\)\\
    \quest & Is \(\sigma\) regret optimal?
  \end{problem}
\end{theorem}
\begin{proof}
  A consequence of the proof of Theorem~\ref{thm:compute} and the existence of
  optipess strategies is that the value \(\Regret{v_0}\) of a game can be
  computed by a polynomial size arithmetic circuit. Moreover, our reliance on
  \(\PosSLP\) allows the input \(r\) Theorem~\ref{thm:compute} to be represented
  as an arithmetic circuit without impacting the complexity.  We can thus verify
  that for all bipositional switching strategies $\sigma'$ (with sufficiently
  large threshold functions) and all possible polynomial size arithmetic
  circuits, \(\Regret{v_0}(\sigma) > r\) implies that \(\Regret{v_0}(\sigma') >
  r\). The latter holds if and only if $\sigma$ is regret optimal since, as
  we have argued in the proof of Theorem~\ref{thm:reg}, such strategies
  $\sigma'$ include optipess strategies and thus regret-minimal
  strategies.\qed
\end{proof}


%% file: conclusions.tex
We studied \emph{regret}, a notion of interest for an agent that does not want
to assume that the environment she plays in is simply adversarial.  We showed
that there are strategies that both minimize regret, and are not consistently
worse than any other strategies.  The problem of computing the minimum regret
value of a game was then explored, and a better algorithm was provided for it.

The exact complexity of this problem remains however open. The only known lower
bound, a straightforward adaptation of~\cite[Lemma 3]{hpr17} for discounted-sum
games, shows that it is at least as hard as solving parity
games~\cite{jurdzinski98}.  Our upper bound could be significantly improved if
we could efficiently solve the following problem:
\begin{problem}
  \given & \((a_i)_{i=1}^n \in \mathbb{Z}^n,\; (b_i)_{i=1}^n \in \mathbb{N}^n\), and \(r
  \in \mathbb{Q}\) all in binary,\\[.7em]
  \quest & Is \(\sum_{i=1}^n a_i\cdot r^{b_i} > 0\)?
\end{problem}
\noindent
The exact complexity of that problem seems to be open even for \(n = 3\).

%% file: appendix.tex

\section{Incomparability of Admissible and Regret-Minimal
Strategies}\label{sec:incomparable}


Let $0 < \lambda <1$. 

Consider the discounted-sum game depicted in Example~\ref{exp:reg_min_not_adm}.
Let $\sigma$ be the strategy of \eve corresponding to the double edges.  This
strategy is \emph{not} admissible: it is dominated by the alternative strategy
$\sigma'$ of \eve that behaves like $\sigma$ from $v_1$ but that chooses to go
to $v'_2$ from $v_2$.  Indeed, if $\tau$ is a strategy of \adam that goes to
$v_1$, then the outcome plays of $\sigma$ and $\sigma'$ are the same, thus have
the same value.  Now, if $\tau$ is a strategy of \adam that goes to $v_2$, then
the value of the outcome play of $\sigma$ and $\tau$ is  $0$, while the value of
the outcome play of $\sigma$ and $\tau$ is $\sum_{i=2}^\infty \lambda^i$ which
is strictly greater than $0$.  However, the strategy $\sigma$ \emph{is}
regret-minimizing: Recall that $\regret{v_0}{\sigma} = \sup_{\tau}
\sup_{\sigma'} \Val(\out{v_0}{\sigma'}{\tau}) - \Val(\out{v_0}{\sigma}{\tau}) $.
If $\tau$ is the strategy of \adam that goes to $v_2$, then the maximal
difference of values between plays following $\sigma$ and plays following
alternative strategies is actually attained with $\sigma'$, and is thus
$\sum_{i=2}^\infty \lambda^i$.  Now, if \adam goes to $v_1$, the maximal
difference of values between plays following $\sigma$ and plays following
alternative strategies is  $\sum_{i=1}^\infty 2\lambda^i$: if the strategy of
\adam is such that it chooses to go to $y$ from $v'_1$, and to $x$ from $v''_1$,
playing $\sigma$ yields a play value of $0$, while going to $v''_1$ yields a
play value of $\sum_{i=2}^\infty 2\lambda^i$, which is strictly greater than
$\sum_{i=2}^\infty \lambda^i$, since $\lambda > 0$.  Thus, we have that
$\regret{v_0}{\sigma} = \sum_{i=2}^\infty 2\lambda^i$.  Symmetrically, any
strategy that chooses to go to $v''_1$ from $v_1$ also has a regret of
$\sum_{i=2}^\infty 2\lambda^i$.  Thus, the strategy $\sigma$ is
regret-minimizing.

Consider now the discounted-sum game depicted in
Example~\ref{exp:adm_not_reg_min}.  Let $\sigma$ be the strategy of \eve
corresponding to the double edges.  This strategy \emph{is} admissible: In this
game, \eve has only two available strategies: $\sigma$ and the strategy
$\sigma'$ that goes to $v_1$ from $v_0$.  It is easy to see that $\sigma$ is not
weakly dominated by $\sigma'$.
Indeed, let us fix the strategy $\tau$ of \adam that goes to
$v''_1$ from $v_1$.  Against $\sigma$, it yields a play value of
$\sum^\infty_{i=2} 6\lambda^i$, while against $\sigma'$, it yields a strictly
smaller play value of $\sum^\infty_{i=2} 5\lambda^i$.  Hence, $\sigma$ is not
dominated by $\sigma'$, and is thus admissible.  The strategy $\sigma$ is
however \emph{not} regret-minimizing.  In fact, the strategy $\sigma'$ has a
smaller regret.  Indeed, the regret of $\sigma$ is the difference between the
best possible outcome value of $\sigma'$, which is $\sum^\infty_{i=2} 10
\lambda^i$ and its own only outcome value $\sum^\infty_{i=2} 6\lambda^i$, that
is, a regret of $\sum^\infty_{i=2} 4\lambda^i$.  On the other hand, the regret
of $\sigma'$ is the difference between the best and only possible outcome value
of $\sigma$, which is $\sum^\infty_{i=2} 6 \lambda^i$ and its own worst possible
outcome value $\sum^\infty_{i=2} 5\lambda^i$, that is, a regret of
$\sum^\infty_{i=2} \lambda^i$, which is strictly less than$\sum^\infty_{i=2}
4\lambda^i$.  Hence, the strategy $\sigma$ is not regret-minimizing. Notice that
the strategy $\sigma'$ is in fact an optipess strategy (even though rather
trivially).

\section{A proof of the existence of SBWO strategies}
We prove the existence of positional SBWO strategies and their
characterization stated in Lemma~\ref{lem:charac-so}.
\begin{proof}
    The characterization of SBWO strategies actually follows directly from the
    characterization of SWO and SBO strategies also given by
    Lemma~\ref{lem:charac-so}, and from the definition of $\acVal$. Below, we
    focus on the positionality claim.

    From~\cite[Theorem 5.1]{zp96} we know that for all $u \in V_\exists$ it
    holds that
    \[
        \aVal^u = \max_{(u,v) \in E} w(u,v) + \lambda \cdot \aVal^v.
    \]
    Denote by $\mathcal{A}$ the game obtained by restricting $\calG$ to the
    subset of edges $E' = \{(u,v) \in E \st u \in V_\exists \implies \aVal^u =
    w(u,v) + \lambda \cdot \aVal^v\}$. It should be clear that $\mathcal{A}$
    characterizes the set of all SWO strategies, i.e. a strategy of \eve is SWO
    in $\calG$ if and only if it is a valid strategy in $\mathcal{A}$. Moreover,
    by definition of $\acVal$, we have that a strategy of \eve in $\calG$ is
    SBWO if and only if it is SBO in $\mathcal{A}$. To conclude, we recall that
    positional SBO strategies for \eve exist in $\mathcal{A}$ (see
    Lemma~\ref{lem:charac-so}, which follows from the corresponding
    \emph{Bellman optimality equations} for $\cVal$ given by~\cite[Theorem
    5.1]{zp96}). \qed
\end{proof}

\section{Proof of Theorem~\ref{thm:regret-optipess}}
It is known that minimal-regret strategies always exist.
\begin{lemma}[Follows from~{\cite[Proposition 18]{hpr16}}]
    \label{lem:reg-min-strats}
    For all games and all initial vertices $v_0$, there exists a strategy
    $\sigma$ of \eve such that $\regret{v_0}{\sigma} = \Regret{v_0}$.
\end{lemma}

The following upper bound on the ``local regret'' 
of strategies that are SWO will be useful.
\begin{lemma}\label{lem:local-upper}
    For all $v_0 \in V_\exists$ and for all SWO strategies $\sigma$ from $v_0$
    we have that
    \[
        \lambda^n\left(
            \cVal^{v_n}_{\lnot \sigma(h)} - \aVal^{v_n}(\sigma_{h})
        \right)
        \leq
        \cVal^{v_0} - \aVal^{v_0}
    \]
    for all histories $h = v_0 \dots v_n$.
\end{lemma}
\begin{proof}
    We first observe that for all strategies $\sigma'$ of \eve and all histories
    $h' = v'_0 \dots v'_m$ we have that $\aVal^{h'}(\sigma') = \aVal^{h'}$ if and
    only if $\aVal^{v'_m}(\sigma'_{h'}) = \aVal^{v'_m}$. Hence, for SWO
    strategies, the latter equality always holds.
    
    The following inequalities yield the result.
    \begin{align*}
        & \cVal^{v_0} - \aVal^{v_0} & \\
        { }\geq{ } & \cVal^{v_0} - \left( \Val(h) + \lambda^n \aVal^{v_n} \right) &
        \text{def. of } \aVal^h\\
        { }\geq{ } & \left(\Val(h) + \lambda^n \cVal^{v_n} \right) -
        \left(\Val(h) + \lambda^n \aVal^{v_n} \right) & 
        \text{def. of } \cVal^h\\
        { }={ } & \lambda^n \left( \cVal^{v_n} - \aVal^{v_n} \right) & \\
        { }={ } & \lambda^n \left( \cVal^{v_n} - \aVal^{v_n}(\sigma_h) \right) &
        \text{by the argument above} \\
        { }\geq{ } & \lambda^n \left( \cVal^{v_n}_{\lnot \sigma(h)} -
        \aVal^{v_n}(\sigma_h) \right) &
        \text{defs. of } \cVal^{v_n}, \cVal^{v_n}_{\lnot \sigma(h)} \text{\qed}
    \end{align*}
\end{proof}

Using the above lemma, it is straightforward to argue
that \eve can switch to follow an SWO strategy without increasing
her regret.
\begin{lemma}\label{lem:simplify}
    Let $\swostrat$ be a SWO strategy of \eve. For all strategies $\sigma$ of
    \eve and all $v_0 \in V$, if we let
    \[
        t(v) = \left\{i \in \mathbb{N} \st \lambda^i \left( \cVal^v - \aVal^v
        \right) \leq \regret{v_0}{\sigma} \right\}
    \]
    for all $v \in V_\exists$
    then $\regret{v_0}{\sigma'} \leq
    \regret{v_0}{\sigma}$ where $\sigma' = \switch{\sigma}{t}{\swostrat}$.
\end{lemma}
\begin{proof}
    Observe that a history consistent with $\sigma'$ is a
    switched history if and only if it has a prefix $v_0 \dots v_n \in
    \histe{v_0}(\sigma')$ such that
    \begin{equation}\label{eqn:cond-simplify}
        \lambda^n\left(
            \cVal^{v_n} - \aVal^{v_n}
        \right)
        \leq
        \regret{v_0}{\sigma}.
    \end{equation}

    Let $S_{\sigma'}$ denote the maximal \emph{local regret} incurred by
    switched histories consistent with $\sigma'$
    and $U$ the maximal local regret incurred by all
    unswitched histories (therefore consistent with both $\sigma$ and $\sigma'$.
    More formally,
    \[
        S_{\sigma'} = \sup_{h'}
        \lambda^n\left(
            \cVal^{v_n}_{\lnot \swostrat(h')} - \aVal^{v_n}({\swostrat_{h'})}
        \right)
    \]
    with the supremum ranging over all switched histories
    $h' = v_0 \dots v_n \in \histe{v_0}(\sigma')$.
    Additionally,
    \[
        U = \sup_{h'}
        \lambda^m\left(
            \cVal^{v'_m}_{\lnot \sigma(h')} - \aVal^{v'_m}({\sigma_{h'})}
        \right)
    \]
    with the supremum ranging over all unswitched histories $h' = \dots v'_m \in
    \histe{v_0}(\sigma) \cap \histe{v_0}(\sigma')$.
    From Lemma~\ref{lem:reg-by-local} and the definition
    of $\sigma'$ it follows that $\regret{v_0}{\sigma'} = \max(S_{\sigma'},U)$.

    Now, consider the value
    \[
        S_0 = \sup_{h'} \lambda^n\left( \cVal^{v_n} - \aVal^{v_n} \right)
    \]
    with the supremum ranging over all switched histories $h' = v_0 \dots v_n
    \in \histe{v_0}(\sigma')$ and such that no proper prefix of $h'$ is a
    switched history. (This indeed implies $h'$ is consistent with $\sigma$
    too.) From Lemma~\ref{lem:local-upper} we have that $S_{\sigma'} \leq S_0$
    and therefore $\regret{v_0}{\sigma'} \leq \max(S_0,U)$. Observe that, using
    Equation~\eqref{eqn:cond-simplify}, we obtain that $S_0 \leq
    \regret{v_0}{\sigma}$. To conclude the proof it thus suffices to show that $U
    \leq \regret{v_0}{\sigma}$.  However, it follows from
    Lemma~\ref{lem:reg-by-local} that $\regret{v_0}{\sigma} = \max(S_\sigma, U)$
    where $S_\sigma$ denotes the maximal local regret incurred by histories
    consistent with $\sigma$ and such that they have a prefix that is a switched
    history consistent with $\sigma'$. Hence, the claim holds.
\qed
\end{proof}

The above result provides us with a way of simplifying regret-minimizing
strategies: For any $v_0 \in V$ and any strategy $\sigma$ of
\eve, there is a second strategy $\sigma'$ of hers that follows $\sigma$ as
long as Equation~\ref{eqn:cond-simplify} does not
hold for the current history.
Otherwise, $\sigma'$ conclusively switches to a worst-case optimal strategy.

The following definition will be useful. We denote by $\cOpt(u)$
the set of all best-case-optimal successors of $u \in V_\exists$, i.e.
\[
    \cOpt(u) \defeq \{v \in V \st (u,v) \in E \text{ and } \cVal^u = \cVal^{uv}\}.
\]
\begin{proof}[of Theorem~\ref{thm:regret-optipess}]
    Lemma~\ref{lem:reg-min-strats} tells us that for all $v_0 \in V$ there
    exists a strategy $\sigma_0$ of \eve such that $\regret{v_0}{\sigma_0} =
    \Regret{v_0}$. Let $\sbwostrat$ be a SBWO strategy of \eve. From
    Lemma~\ref{lem:simplify} we get that the strategy $\sigma =
    \switch{\sigma_0}{t}{\sbwostrat}$, where for all $v \in V_\exists$ we have
    \[
        t(v) = 
        \left\{ i \in \mathbb{N} \st \lambda^i \left( \cVal^v - \aVal^v \right)
        \leq \Regret{v_0} \right\},
    \]
    is also such that $\regret{v_0}{\sigma} = \Regret{v_0}$. We will now
    argue that $\sigma$ is an optipess strategy. In fact, we will prove
    something slightly stronger: for all SBO strategies $\sbostrat$ of \eve, for
    all $h = v_0 \dots v_n \in \histe{v_0}(\sigma)$ such that $h$ is an
    unswitched history, we have that
    \begin{enumerate}
        \item $|\cOpt(v_n)| = 1$ and
        \item $\sigma(h) = \sbostrat(h)$.
    \end{enumerate}
    Towards a contradiction, assume that this is not the case. That is, there
    exists such an $h$ for which $|\cOpt(v_n)| > 1$ or 
    $|\cOpt(v_n)| = 1$ but $\sigma(h) \neq
    \sbostrat(h)$ for all SBO strategies $\sbostrat$ of \eve. In the latter case,
    by Lemma~\ref{lem:charac-so}, we must have that $\sigma(h) \not\in
    \cOpt(v_n)$. It should be clear that in either case we have that
    $\cVal_{\lnot \sigma(h)}^{v_n} = \cVal^{v_n}$. We thus have that
    \begin{align*}
        \lambda^{-n} \regret{v_0}{\sigma}
        & \geq \cVal^{v_n}_{\lnot \sigma(h)} - \aVal^{v_n}(\sigma_h) &
        \text{by Lemma~\ref{lem:reg-by-local}}\\
        & = \cVal^{v_n} - \aVal^{v_n}(\sigma_h) & \text{see arguments above}\\
        & \geq \cVal^{v_n} - \aVal^{v_n} & \text{by definition of } \aVal.
    \end{align*}
    By assumption, we have that $h$ is an unswitched history and therefore
    \[
        \lambda^n\left(\cVal^{v_n} - \aVal^{v_n}\right)
        > \Regret{v_0}.
    \]
    The above inequalities thus imply that the regret of $\sigma$ is
    strictly larger than $\Regret{v_0}$, which is a contradiction.\qed
\end{proof}

\section{Proof of Theorem~\ref{thm:fundamental-property}}
\begin{proof}[of Theorem~\ref{thm:fundamental-property}]
    Let $\sigma$ be a strategy of \eve and $\mathcal{D}$ be the set of
    histories $h$ such that the sequence of inequalities $\aVal^{h}(\sigma) \leq
    \cVal^{h}(\sigma) \leq \aVal^{h} \leq  \acVal^{h}$ holds with at least one
    inequality being strict. Denote by $sp(\mathcal{D})$ be the (possibly
    infinite) subset of $\mathcal{D}$ that contains all the shortest prefixes of
    the histories in $\mathcal{D}$, that is
    \[
        sp(\mathcal{D}) \defeq \left\lbrace h \in \mathcal{D} 
        \:\middle|\: \forall h' \in \mathcal{D} \setminus \{h\}:
        h' \not\prefeq h \right\rbrace
    \] 

    We now define a strategy $\sigma'$ for all histories $h = v_0 \dots v_n$ such
    that $v_n \in V_\exists$ as follows
    \[
        \sigma'(h) = \begin{cases}
            \sbwostrat_{h'}(h) & \text{if there exists } h' \in sp(\mathcal{D})
            \text{ such that } h' \prefeq h\\
            \sigma(h) & \text{otherwise.}
        \end{cases}
    \]
    (Note that $\sigma'$ is well-defined as all the elements in
    $sp(\mathcal{D})$ are incomparable.)
    
\medskip    
    
    Intuitively, the strategy $\sigma'$ follows $\sigma$ until the above
    sequence of inequalities holds --- which, essentially,
    means that one can do better than $\sigma$ from that point onward. Then,
    $\sigma'$ switches to follow an SBWO strategy forever.

    We first show that $\sigma$ is weakly dominated by $\sigma'$.  To do so, we
    will compare $\sigma$ and $\sigma'$ with regard to the strategies of \adam.
    Let $\tau$ be a strategy of \adam.  Let $\pi$ be the outcome play of
    $\sigma$ and $\tau$, and $\pi'$ the one of $\sigma'$ and $\tau$.  If $\pi=
    \pi'$, then clearly $\Val(\pi) = \Val(\pi')$.  Otherwise, if $\pi \neq
    \pi'$, they share a longest common prefix $h= v_0\dots v_n$.  As $\tau$ is
    fixed, we know that $v_n \in V_\exists$ and that $\sigma(h) \neq
    \sigma'(h)$.  By definition of $\sigma'$, it means that there exists a
    prefix $h'$ of $h$ such that $h' \in sp(\mathcal{D})$.  Thus, we have that
    $\cVal^{h'}(\sigma) \leq \aVal^{h'}$ and consequently $\Val(\pi) \leq
    \aVal^{h'}$.  On the other hand, from $h'$ we know that $\sigma'$ behaves
    like $\sbwostrat_{h'}$, thus, $\Val(\pi') \geq \aVal^{h'}$.  Hence,
    $\Val(\pi) \leq \Val(\pi')$.  This is true for any strategy of \adam, thus
    $\sigma$ is indeed weakly dominated by $\sigma'$.

\medskip

    We now show that $\sigma'$ is admissible.  Towards this, we use the
    characterization from Lemma~\ref{lem:adm_char}: \emph{A strategy $\sigma$
    of \eve is admissible if and only if for every history $h \in
    \histe{v_0}(\sigma)$ the following holds: either $\cVal^{h}(\sigma) >
    \aVal^h$ or $\aVal^h(\sigma)=\cVal^h(\sigma)=\aVal^h = \acVal^h$.} Let
    $h=v_0\dots v_n$ be a history consistent with $\sigma'$ such that $v_n \in
    V_{\exists}$.
    \begin{itemize}
        \item Assume first that there exists a prefix $h'$ of $h$ that belongs
            to $sp(\mathcal{D})$.  In that case, we know that $\sigma'$ behaves
            like $\sbwostrat_{h'}$ from $h$, thus we have $\aVal^h(\sigma')=\aVal^h$
            and $\cVal^h(\sigma')=\acVal^h$, by definition of SBWO strategies.
            If \[\acVal^h > \aVal^h,\] then
            $\sigma'$ satisfies the first part of the characterization.
            Otherwise, $\acVal^h = \aVal^h$, the second part of the
            characterization is satisfied, as we obtain immediately
            $\aVal^h(\sigma')=\cVal^h(\sigma')=\aVal^h=\acVal^h$.
            \medskip
            
        \item Assume now that $h$ has no prefix that belongs to
            $sp(\mathcal{D})$.  By definition of $\sigma'$, this means that
            $\sigma'(h') = \sigma(h)$ for all prefixes $h' \prefeq h$.  In other
            terms, $\sigma$ and $\sigma'$ agree (at least) up to $h$.  Let
            $h\pi$ be an outcome consistent with $\sigma$ such that $\Val(h\pi)=
            \cVal^h(\sigma)$ (which exists because discounted-sum games are
            well-formed).  Let $\tau$ be a strategy of \adam such that
            $\pi^{v_0}_{\sigma\tau} = h\pi$ (which exists because $h\pi$ is
            consistent with $\sigma$).  Since $\sigma$ and $\sigma'$ agree up to
            $h$, there exists $\pi'$ be such that $h\pi' =
            \pi^{v_0}_{\sigma'\tau}$.  Recall that $\sigma$ is weakly dominated
            by $\sigma'$.  As $\tau$ is fixed, we know that $\Val(h\pi) \leq
            \Val(h\pi')$.  Thus, we have that $\cVal^h(\sigma') \geq \Val(h\pi')
            \geq \Val(h\pi) = \cVal^h(\sigma)$.
                        
            Recall now that by definition of $sp(\mathcal{D})$, we know that in
            particular, it either holds that \[(A)~ \cVal^h(\sigma) > \aVal^h\] or
            \[(B)~ \aVal^h(\sigma) = \cVal^h(\sigma)= \aVal^h = \acVal^h.\]
            Suppose $(A)$ holds.  
            We thus have that $\cVal^h(\sigma') \geq \cVal^h(\sigma) > \aVal^h$,
            that is, $\cVal^h(\sigma') > \aVal^h$.  This means that $\sigma'$
            satisfies the first part of the characterization.  
            
            Finally, suppose that $(B)$ holds.
            We have 
            \begin{align*}
                & \cVal^h(\sigma') \geq \Val(h\pi') \geq \Val(h\pi)\\
                {} = {} & \cVal^h(\sigma) = \aVal^h(\sigma) = \aVal^h = \acVal^h,
            \end{align*}
            and thus
            $\cVal^h(\sigma') \geq \acVal^h$ .  Furthermore, we also know that
            $\aVal^h(\sigma') \geq \aVal^h(\sigma)$.  As by definition of the
            antagonistic value, we have $\aVal^h(\sigma') \leq \aVal^h$ and
            $\aVal^h(\sigma) = \aVal^h$, we obtain $\aVal^h(\sigma') = \aVal^h$.
            We now know that $\sigma'$ is worst-case optimal at $h$.  By
            definition of $\acVal^h$, we can conclude that $\cVal^h(\sigma')
            \leq \acVal^h$.  Since it is also true that $\cVal^h(\sigma') \geq
            \acVal^h$, we obtain $\aVal^h(\sigma') = \cVal^h(\sigma')= \aVal^h =
            \acVal^h$, that is, $\sigma'$ satisfies the second part of the
            characterization.
    \end{itemize}
    Thus, the strategy $\sigma'$ is admissible.\qed
\end{proof}

\section{On the well-formedness of discounted-sum games}\label{sec:well-formed}
In~\cite{bprs16}, the authors introduce the notion of \emph{well-formed} games,
that is, games where, for each player, and each history, there exist strategies
witnessing the antagonistic and collaborative values at this history.  They then
show that, in such games, admissible strategies can be characterized in terms of
values at any history consistent with the strategy (see
Lemma~\ref{lem:adm_char}).  It is worth noticing that for any player, it is
in fact sufficient that this player has witnessing strategies for the
antagonistic and collaborative values at any history.  We call this property
\emph{well-formedness for a player}.  In our context, we focus on the strategies
and payoffs of \eve, thus we phrase the statement as follows:

\medskip
A game is well-formed \emph{for \eve} if, for all $h \in \histe{v_0}$:
\begin{enumerate}
\item there exists a strategy $\sigma$ of \eve such that $\aVal^h(\sigma)= \aVal^h$. 
\item there exists a strategy $\sigma$ of \eve such that $\cVal^h(\sigma)= \cVal^h$. 
\end{enumerate}

\begin{lemma}\label{lem:disc-sum_well_formed}
Discounted-sum games are well-formed for \eve.
\end{lemma}

\begin{proof}
This can be seen as a direct implication of Lemma~\ref{lem:charac-so}: Indeed,
the SWO and SBWO strategies are good witnesses for conditions $1$ and $2$,
respectively. 
\qed \end{proof}

Lemma~\ref{lem:adm_char} then directly follows from~\cite[Theorem 11]{bprs16}.

\medskip

Note that well-formedness for \eve, in general, does not guarantee the existence
of a play that witnesses the collaborative value at any history.  However, in
discounted-sum games, this is indeed the case.  In the proof of
Theorem~\ref{thm:optipess_adm}, we use the fact that there exists a play
consistent with $\sbostrat$ that has such value, thus also a strategy $\tau$ of
\adam such that the outcome of $\sbostrat$ and $\tau$ is exactly this play.  The
argument relies on the fact that discounted-sum value functions are
\emph{continuous}.  We recall a few useful notions before proving the property.

\medskip
Considering a discounted-sum game $\calG= (V,v_0, V_\exists,E,w,\lambda)$.  The
set $V$ is endowed with the discrete topology, and thus the set $V^\omega$ with
the product topology.  Then, a sequence of plays $(\pi_n)_{n \in \mathbb{N}}$
is said to converge to a play $\pi = \lim_{n \to \infty} \pi_n$, if every
prefix of $\pi$ is a prefix of all but finitely many of the $\pi_n$.
It is well known that the discounted-sum value function is
continuous, that is, whenever a sequence of plays $(\pi_n)_{n \in \mathbb{N}}$
converges to a play $\pi$, we have $\lim_{n \to \infty} \Val(\pi_n) =
\Val(\pi)$.

\begin{lemma}
    For any history $h=v_0 \dots v_n$ consistent with $\sbostrat$, there exists
    a strategy $\tau$ of \adam such that $h \prefeq \out{v_0}{\sbostrat}{\tau}$
    and $\Val(\out{v_0}{\sbostrat}{\tau})= \cVal^h$.
\end{lemma}

\begin{proof}

From Lemma~\ref{lem:charac-so}, we know that $\cVal^h(\sbostrat)=\cVal^h$.
Thus, we have, by definition of the collaborative value, that
\[
    \cVal^h = \Val(h) + \lambda^{|h|} \cdot \sup_\tau
    \Val(\out{v_n}{\sbostrat_h}{\tau}).
\]
As $\sbostrat$ is positional, we have $\sbostrat_h = \sbostrat$, thus we can
write 
\[
    \cVal^h = \Val(h) + \lambda^{|h|} \cdot \sup_\tau
    \Val(\out{v_n}{\sbostrat}{\tau}).
\]
Let $d \defeq \sup_\tau \Val(\out{v_n}{\sbostrat}{\tau})$.  We first show that
there exists $\tau'$ such that $\Val(\out{v_n}{\sbostrat}{\tau'})=d$.

\medskip

Since $d = \sup_\tau \Val(\out{v_n}{\sbostrat}{\tau})$, we know that there
exists a sequence $(\pi_n)_{n\in \mathbb{N}}$ of outcomes consistent with
$\sbostrat$ such that $\limsup_{n \to \infty} \Val(\pi_n) = d$.  Since the
discounted-sum value function is continuous, we also have that $\lim_{n \to
\infty} \Val(\pi_n) = d$.

\medskip

Suppose the sequence $(\pi_n)_{n\in \mathbb{N}}$ eventually stabilizes, that
is, there exists $N$ such that $\pi_n = \pi_N$ for every $n > N$.  We have
that $\lim_{n \to \infty} \Val(\pi_n) = \Val(\pi_N)$, thus $\Val(\pi_N) = d$.
Let $\tau'$ be a strategy of \adam such that $\out{v_n}{\sbostrat}{\tau'} =
\pi_N$ (which exists since $\pi_N$ is a valid outcome in $\calG$).  We indeed
obtain that $\Val(\out{v_n}{\sbostrat}{\tau'}) = d$.

\medskip

Suppose now the sequence $(\pi_n)_{n\in \mathbb{N}}$ does not eventually
stabilize, that is, for all $n$, there exists $N>n$ such that $\pi_N \neq
\pi_n$.  We construct, iteratively, a subsequence $(\pi'_k)_{k\in \mathbb{N}}$
from $(\pi_n)_{n\in \mathbb{N}}$ as follows: We start by fixing $\pi'_0 =
\pi_0$.  Recall that $V$ is a finite set.  Let $m$ be its size, we can label
the vertices $v^0, \dots, v^{m-1}$, with $v^0=v_0$.  Let $P_0 \defeq \lbrace \pi_n
\st n\in\mathbb{N} \rbrace$.  We partition the set of all $\pi_n$ according to
their prefixes of length $2$: For every $0 \leq i < m$, we define $P^i_1 \defeq
\lbrace \pi_n \st n \in \mathbb{N}, ~ v_0v^i \prefeq \pi_n \rbrace$.  As $V$
is finite and the set of outcomes in the sequence is infinite, there exists $i$
such that the set $P^i_1$ is infinite as well.  We fix $P_1$ to be such an infinite
$P^i_1$.  Let $\pi \in P_1$.  We fix $\pi'_1 = \pi$.

Suppose now that $\pi'_0$ to $\pi'_k$ are already determined, as well as the
infinite sets $P_0$ to $P_k$, and that  $v_0 \dots v'_k$ is the prefix of length
$k+1$ of $\pi'_k$.  For every $0 \leq i < m$, we define $P^i_{k+1} \defeq \lbrace
\pi_n \st \pi_n \in P_k, ~ v_0 \dots v'_k v^i \prefeq \pi_n \rbrace$.  Again,
as $V$ is finite and the set $P_k$ is infinite, there exists $i$ such that the
set $P^i_{k+1}$ is infinite as well.  Let $P_{k+1} \defeq P^i_{k+1}$.  Let $\pi \in
P_{k+1}$.  We fix $\pi'_{k+1} = \pi$.

The subsequence $(\pi'_k)_{k\in \mathbb{N}}$ is now well defined and has the
following property: for each $N \in \mathbb{N}$, each prefix $h_N$ of length
$N+1$ of $\pi'_N$, and $k \geq N$, we have $h_N \prefeq \pi'_k$.  Let $\pi$
be the outcome such that $h_N \prefeq \pi$ for all $N \in \mathbb{N}$ (this
outcome is well defined as $h_N \prefeq h_{N+1}$ for all $N \in \mathbb{N}$).
By construction of $\pi$, we have that $\lim_{k \to \infty} \pi'_k = \pi$.
As $(\pi'_k)_{k\in \mathbb{N}}$ is a subsequence of $(\pi_n)_{n\in
\mathbb{N}}$, the sequence $(\Val(\pi'_k))_{k\in \mathbb{N}}$ is a
subsequence of $(\Val(\pi_n))_{n\in \mathbb{N}}$, hence:
\[
    \lim_{k\to \infty} \Val(\pi'_k) = \lim_{n \to \infty} \Val(\pi_n) = d.
\]

Since the discounted-sum value function is continuous, we also have $\Val(\pi)=
\lim_{k \to \infty} \Val(\pi'_k)$.  Thus we get $\Val(\pi) = d$.  Let $\tau'$
be a strategy of \adam such that $\out{v_n}{\sbostrat}{\tau'} = \pi$ (which
exists since $\pi$ is a valid outcome in $\calG$).  We indeed obtain that
$\Val(\out{v_n}{\sbostrat}{\tau'}) = d$.

\medskip

We now conclude the proof by exhibiting $\tau$ such that
$\Val(\out{v_0}{\sbostrat}{\tau})= \cVal^h$: We already know that $h$ is
consistent with $\sbostrat$.  This means in particular that there exists a
strategy $\tau''$ of \adam such that $h \prefeq \out{v_0}{\sbostrat}{\tau''}$.
Let now $\tau$ be the strategy of \adam such that:
\[
    \tau(h') \defeq \begin{cases}
        \tau'(h') &\text{if } h \prefeq h',\\
        \tau''(h') & \text{otherwise.}
    \end{cases}
\]
It is easy to see that $\out{v_0}{\sbostrat}{\tau} = h\pi$.  Finally,
$\Val(h\pi)= \Val(h) + \lambda^{|h|} \cdot \Val(\pi) = \Val(h) + \lambda^{|h|}
\cdot d = \sup_\tau \Val(\out{v_n}{\sbostrat}{\tau}) =\cVal^h$.
\qed \end{proof}

\section{Proof of Lemma~\ref{lem:replace-cycle}}
\begin{proof}[of Lemma~\ref{lem:replace-cycle}]
    We will suppose that neither inequality holds and derive a
    contradiction.

    Let $k$ and $\ell$ be the lengths of $\alpha,\gamma$ and $\beta$,
    respectively.  On the one hand, we have that $\Val(\alpha^2 \cdot \beta) >
    \Val(P)$. This is equivalent to the following.
    \begin{align}
        & \Val(\alpha^2) + \lambda^{2k}\Val(\beta) > \Val(\alpha) + \lambda^k
        \Val(\beta) + \lambda^{k + \ell} \Val(\gamma) \nonumber \\
        \iff & \lambda^k \Val(\alpha) + \lambda^{2k} \Val(\beta) >
        \lambda^k \Val(\beta) + \lambda^{k + \ell} \Val(\gamma) \nonumber \\
        \iff & \Val(\alpha) + \lambda^k \Val(\beta) > \Val(\beta) + \lambda^\ell
        \Val(\gamma) \nonumber \\
        \iff & \Val(\alpha) > (1-\lambda^k)\Val(\beta) + \lambda^\ell
        \Val(\gamma). \label{ineq:first}
    \end{align}
    On the other hand, we have that $\Val(\beta \cdot \gamma^2) > \Val(P)$. The
    latter holds if and only if the following does.
    \begin{align*}
        & \Val(\beta) + \lambda^\ell \Val(\gamma)^2 > \Val(\alpha) + \lambda^k
        \Val(\beta) + \lambda^{k + \ell} \Val(\gamma)\\
        \iff & (1-\lambda^k) \Val(\beta) + \lambda^\ell \Val(\gamma^2) >
        \Val(\alpha) + \lambda^{k + \ell} \Val(\gamma)\\
        \iff & (1 - \lambda^k) \Val(\beta) + \lambda^\ell \Val(\gamma) +
        \lambda^{k + \ell} \Val(\gamma) > \Val(\alpha) + \lambda^{k + \ell}
        \Val(\gamma)\\
        \iff & (1 - \lambda^k) \Val(\beta) + \lambda^\ell \Val(\gamma) >
        \Val(\alpha).
    \end{align*}
    The last inequality is already in clear contradiction with
    Inequality~\ref{ineq:first}.\qed
\end{proof}

\section{Representing and comparing long-history values}\label{sec:posslp}
Presently, we provide a brief
discussion on succinctly encoded (rational) numbers. In this work we have
assumed that all weights labeling edges in our game are given as
binary-encoded numbers.  The discount factor, $\lambda$, we also assume is given
in binary.  That is, $\lambda$ is given as a pair of binary-encoded natural
numbers $p,q \in \mathbb{N}$ such that $q > 0$ and $p/q = \lambda$. In
Section~\ref{sec:algos} we deal with numbers that seemingly do not admit such
classical representations.

Besides encoding a number in binary, one can also consider polynomials (and a
binary-encoded valuation of its variables), or \emph{arithmetic circuits} as
representations for numbers (see, e.g.,~\cite{ab09}). A number $P =
a_n^{e_1} + \dots + a_n^{e_n}$ where $a_i \in \mathbb{Z}$ and $e_i \in
\mathbb{N}$ for all $1 \leq i \leq n$, for instance, may be such that $P(
\vec{a}, \vec{e}) \geq 2^{2^n}$ while being representable with a list of
binary-encoded numbers using at most $n^2$ bits. An arithmetic circuit is an
even more succinct representation. Formally, such a circuit is a rooted directed
acyclic graph whose internal nodes are labelled with operations from
$\{+,-,\times\}$ and whose leaves are labelled with binary-encoded integers.
Determining whether a number given as an arithmetic circuit is positive is known
as the \PosSLP~problem and has been shown to be decidable in the fourth level of
the counting hierarchy by Allender et al.~\cite{abkm09}.

In this work, because of the discount factor, when writing formulas for the
discounted-sum value of long histories, we may in fact need to use division.
Concretely, to determine whether the value of a history is positive one may write
down the following inequality
\[
    \sum^m_{i=1} \left(\frac{p}{q}\right)^{e_i} \left(\frac{a_i}{b_i}\right) > 0
\]
where $b_i,e_i \in \mathbb{N}$, $q_i > 0$, and $a_i \in \mathbb{Z}$ for all $1
\leq i \leq m$. However, we can remove this limited use of division by doing the
following. Let $E \defeq \max\{ e_i \st 1 \leq i \leq m\}$ and $B_i \defeq
\prod\{b_j \st 1 \leq j \leq m, j \neq i\}$. Then the above inequality
holds if and only if the following holds
\[
    \sum^m_{i=1} B_i b_i p^{e_i} q^{E - e_i} > 0.
\]

In the context of this work, 
the main application of the arithmetic-circuit-encoding discussed here is
to express the discounted-sum value of a long history $\alpha \cdot \beta^k
\cdot \gamma$ as follows
\[
    \Val(\alpha\cdot\beta^k\cdot\gamma)
    = \Val(\alpha) + \lambda^{|\alpha|} \frac{\Val(\beta)}{1 -
        \lambda^{|\beta|}}(1 -
    \lambda^{|\beta|k}) + \lambda^{|\alpha|+|\beta|k}\Val(\gamma)
\]
Note that while $\Val(\alpha)$,
$\Val(\beta)$, and $\Val(\gamma)$ can be represented using binary rationals,
this is not the case for $\lambda^{|\beta|k}$ in general.

\section{Upper-bounding \(t(v)\) for optipess strategies}\label{sec:boundopti}
We will now prove the following bound on the finite values of the threshold
function for optipess strategies:
\begin{lemma}\label{lem:bound-threshold-optipess}
    For all optipess strategies $\sigma$ of \eve with
    threshold function $t$ we have that $t(v)$ is at most exponential
    for all $v \in V_\exists$ with $t(v) \in \mathbb{N}$.
\end{lemma}

Let us fix a value for the \emph{size of a game} with discount factor $\lambda =
p/q$. Define
\[
    |\calG| \defeq |V| + |E| + \lceil \log_2 p \rceil + \lceil \log_2 q \rceil +
    \sum_{(u,v) \in E} \lceil \log_2 w(u,v) \rceil.
\]

\subsection{A lower bound on the regret of a game.}
In~\cite{hpr16} the following lower bound on the regret of games with non-zero
regret was given.
\begin{lemma}[From~{\cite[Corollary 12]{hpr16}}]
    For all $v_0 \in V$ we have that
    if $\Regret{v_0} > 0$ then 
    \[
        \Regret{v_0} \geq \min\left\{ \lambda^{|V|} \left( \cVal^v -
        \aVal^v \right) \:\middle|\: v \in V_\exists, \cVal^v > \aVal^v \right\}.
    \]
\end{lemma}

Using the existence of positional optimal strategies in discounted-sum games
(see Lemma~\ref{lem:charac-so}) it is straightforwards to show the antagonistic
and collaborative values are always realized by a \emph{simple lasso}. That is, a
play $\alpha \cdot \gamma^\omega$ where $\alpha$ is a simple path and $\gamma$
is a simple cycle. It follows that both values are representable using
binary-number pairs that use polynomially-many bits.
\begin{lemma}
    There exists a polynomial $P$ such that for all $v \in V$ we have that
    \begin{itemize}
        \item $\aVal^v = a/b$, $\cVal^v = c/d$, and
        \item $\lceil \log_2 |a| \rceil, \lceil \log_2 |c| \rceil,
               \lceil \log_2 b \rceil, \lceil \log_2 d \rceil \in
               \mathcal{O}(P(|\calG|))$
    \end{itemize}    
    for some $a,c \in \mathbb{Z}$ and $b,d \in \mathbb{N}_{>0}$.
\end{lemma}

As an immediate consequence of the above lemmas we get the following lower bound
on non-zero regret values.
\begin{proposition}\label{pro:lower-bound}
    There exists a polynomial $P$ such that for all $v_0 \in V$ we have that
    if $\Regret{v_0} > 0$ then $\Regret{v_0} \geq 2^{-P(|G|)}$.
\end{proposition}

\subsection{An upper bound on the finite thresholds of an optipess strategy}
We first note that, for all $v \in V_\exists$, if $\cVal^v = \aVal^v$ then $t(v)
= 0$ and if $\Regret{v_0} = 0$ then $t(v) = \infty$. Hence, it suffices to bound
the threshold function for all $v \in V_\exists$ such that $\cVal^v > \aVal^v$
when $\Regret{v_0} > 0$. In the sequel we focus on an arbitrary vertex $v \in
V_\exists$ and make the assumption that those two inequalities hold.

Observe that for all $v \in V$ and all $i,r \in \mathbb{Q}_{\geq 0}$ we have
that 
\[
    \lambda^i \left( \cVal^v - \aVal^v \right) = r
\]
if and only if $i \log_2
\lambda = \log_2 r - \log_2 \left( \cVal^v - \aVal^v \right)$ if and only if
\begin{align*}
    i = {} & \frac{\log_2 r - \log_2 \left( \cVal^v - \aVal^v \right)}{\log_2
        \lambda} \\
    {} = {} & \frac{\log_2 \left( \cVal^v - \aVal^v \right) - \log_2 r}{\log_2
        \left(\lambda^{-1}\right)}.
\end{align*}

Let $\wmax \defeq \max_{(u,u')} |w(u,u')|$. It is easy to see that
\[
    \frac{-\wmax}{1-\lambda} \leq \aVal^v \leq \cVal^v \leq
    \frac{\wmax}{1-\lambda}
\]
for all $v \in V$. From the above arguments we therefore get that for all $v_0
\in V$ the following hold
\begin{align*}
    t(v) = {} & \inf \left\{ n \in \mathbb{N} \:\middle|\: \lambda^n \left(
            \cVal^v - \aVal^v \right) \leq \Regret{v_0} \right\}\\
    {} = {} & \min \left\{ n \in \mathbb{N} \:\middle|\: \lambda^n \left(
            \cVal^v - \aVal^v \right) \leq \Regret{v_0} \right\}\\
    {} \leq {} & \frac{\log_2(\cVal^v-\aVal^v) -
        \log_2 \Regret{v_0}}{\log_2\left(\lambda^{-1}\right)}\\
        {} \leq {} & \frac{\log_2\left(\frac{2\wmax}{1-\lambda}\right) +
        P(|\calG|)}{\log_2\left(\lambda^{-1}\right)}\\
    {} \leq {} & \frac{\log_2(2) + \log_2 \wmax - \log_2(1-\lambda) +
        P(|\calG|)}{\log_2\left(\lambda^{-1}\right)}\\
    {} \leq {} & \frac{1 + \log_2 \wmax + \log_2\left(\frac{q}{q-p}\right) +
        P(|\calG|)}{\log_2(q/p)}\\
    {} \leq {} & \frac{1 + \log_2 \wmax + \log_2(q) +
        P(|\calG|)}{\log_2(q/p)}\\
    {} \leq {} & \frac{P(|\calG|) + 2|\calG| + 1}{\log_2(q/p)}
\end{align*}
where $\lambda = p/q$ and $P$ is the polynomial from
Proposition~\ref{pro:lower-bound}. 

To complete the proof of the claim, it suffices to argue that $1/\log_2
\left( q/p \right)$ grows at most exponentially in the size of $\calG$. It
should be clear that $1/\log_2 \left(q/p \right)$ is maximized when $q/p$
approaches one and that therefore an exponential bound for $1/\log_2 \left(1 +
2^{-|\calG|}\right)$ implies the desired result. Finally, it is easy to verify
that
\[
    \lim_{x \to \infty} \frac{\frac{1}{\log_2 \left(1 + 2^{-x}
        \right)}}{2^{x^2}} = 0,
\]
which implies
\[
    \frac{1}{\log_2 \left(1 + 2^{-|\calG|}\right)}
    \in \mathcal{O}\left(2^{|\calG|^2}\right) \implies
    t(v) \in \mathcal{O}\left(2^{|\calG|^2} (P(|\calG|) + 2|\calG| + 1) \right)
\]
thus completing the proof. \qed